\documentclass[sigconf]{acmart}

\AtBeginDocument{%
  \providecommand\BibTeX{{%
    \normalfont B\kern-0.5em{\scshape i\kern-0.25em b}\kern-0.8em\TeX}}}

\usepackage{amsmath}
\usepackage{amsthm}
\usepackage{amsfonts,dsfont}
\usepackage{bm}
\usepackage[linesnumbered, boxed, ruled, commentsnumbered]{algorithm2e}
\usepackage{url}
\usepackage{multirow}
\usepackage{booktabs}
\newcommand{\tabincell}[2]{\begin{tabular}{@{}#1@{}}#2\end{tabular}}
\usepackage[caption=false,font=footnotesize,labelfont=sf,textfont=sf]{subfig}
\usepackage{hyperref}

\newtheorem{prop}{Proposition}

\begin{document}

\title{Adversarial Representation Sharing: A Quantitative and Secure Collaborative Learning Framework}

\author{Jikun Chen}
\email{cjk7989@sjtu.edu.cn}
\affiliation{%
	\institution{Shanghai Jiao Tong University}
}

\author{Feng Qiang}
\email{f.qiang@gmail.com}
\affiliation{%
	\institution{ }
}

\author{Na Ruan}
\authornote{Corresponding author.}
\email{naruan@cs.sjtu.edu.cn}
\affiliation{%
  \institution{Shanghai Jiao Tong University}
}

\settopmatter{printacmref=false} 
\renewcommand\footnotetextcopyrightpermission[1]{} 
\pagestyle{plain} 

\begin{abstract}
  The performance of deep learning models highly depends on the amount of training data. It is common practice for today's data holders to merge their datasets and train models collaboratively, which yet poses a threat to data privacy. Different from existing methods such as secure multi-party computation (MPC) and federated learning (FL), we find representation learning has unique advantages in collaborative learning due to the lower communication overhead and task-independency. However, data representations face the threat of model inversion attacks. In this article, we formally define the collaborative learning scenario, and quantify data utility and privacy. Then we present ARS, a collaborative learning framework wherein users share representations of data to train models, and add imperceptible adversarial noise to data representations against reconstruction or attribute extraction attacks. By evaluating ARS in different contexts, we demonstrate that our mechanism is effective against model inversion attacks, and achieves a balance between privacy and utility. The ARS framework has wide applicability. First, ARS is valid for various data types, not limited to images. Second, data representations shared by users can be utilized in different tasks. Third, the framework can be easily extended to the vertical data partitioning scenario.
\end{abstract}


\maketitle

\section{Introduction}
\label{sect:introduction}
Deep learning has made great progress in a variety of fields, such as computer vision, natural language processing and recommendation systems. This impressive success largely attributes to the increasing amount of available computation and datasets \cite{goodfellow2016deep}. Companies and institutions require large-scale data to build stronger machine learning systems, whereas data is generally held by distributed parties. Therefore, it is a common practice for multiple parties to share data and train deep learning models collaboratively \cite{ohrimenko2016oblivious}. In most of collaborative training scenarios, users provide their local data to cloud computing services or share data with others, which brings privacy concerns.

Due to security considerations and privacy protection regulations, it is inappropriate to exchange data among different organizations. Obviously, sharing raw data directly may cause a leakage of private and sensitive information contained in datasets. For instance, if some hospitals hope to integrate their patients' information to establish models for disease detection, they must carefully protect the identity of patients from being obtained and abused by any partner or possible eavesdropper. The problem of "data islands" requires a privacy-preserving collaborative framework.

Generally speaking, privacy-preserving collaborative learning can be grouped under two approaches. Earlier works focus on MPC, which ensures that multiple users can jointly calculate a certain function while keeping their inputs secret without the trusted third party. In a joint learning framework, data in local devices are encrypted by cryptographic tools before being shared \cite{mohassel2017secureml,agrawal2019quotient}. Cryptographic methods contain garbled circuits (GC) \cite{yao1986generate}, secret sharing (SS) \cite{paillier1999public}, homomorphic encryption (HE) \cite{gentry2009fully}, etc. However, current cryptographic approaches can just perform several types of operations, and only propose friendly alternatives to some of non-linear functions \cite{mohassel2017secureml}. Moreover, encryption often causes high computation and communication cost.

Shokri et al. \cite{shokri2015privacy} proposed a distributed stochastic gradient descent algorithm to replace the sharing training data framework. The method is now well known as federated learning (FL) \cite{hard2018federated} and has a wide application in practice. In FL, a cloud server builds a global deep learning model. For each training iteration, the server randomly sends the model to a part of client devices. The clients then optimize the model locally, and send the updates back to the server to aggregate them. Only parameters of the model are communicated, while the training data is retained by the local device, which ensures the privacy. Some recent works combine federated learning with other information security mechanisms (e.g., differential privacy) to further improve privacy \cite{geyer2017differentially}. However, the communication cost between each local device and the central server is high. After each iteration of training process, each user needs to keep their local deep learning model synchronized.

Different from the above approaches, we consider representation learning \cite{bengio2013representation} to solve this problem. The idea is inspired by deep neural networks, which embed inputs into real feature vectors (representations). Containing high-level features of the original data, latent representations is efficient to various downstream machine learning tasks \cite{goodfellow2016deep}. The motivation is that both MPC and FL conduct collaborative learning with limited task-applicability. Once the machine learning task changes, the entire training process needs to be executed again, which incurs high communication cost. Comparatively speaking, data representations are task-independent and thus has unique potential in joint learning. Some recent works have studied privacy-preserving data representations \cite{xiao2019adversarial,hitaj2017deep}, but few of them gave further discussion on the collaborative learning scenario. The primary problem in privacy representations learning field is to defend against model inversion attacks\cite{mahendran2015understanding,he2019model}, which aims to train inverse models to reconstruct original inputs or extract private attributes from shared data representations.

In this article, we propose \textbf{\textit{ARS}} (for Adversarial Representation Sharing), a decentralized collaborative learning framework based on data representations. Our work contains two levels: (i) a collaborative learning framework wherein users share data representations instead of raw data for further training; and (ii) an imperceptible adversarial noise added to shared data representations to defend against model inversion attacks. ARS helps joint learning participants "encode" their data locally, then add adversarial noise to representations before sharing them. The published \textit{adversarial latent representation} can defend against reconstruction attacks, thereby ensuring privacy.

Owning to the good qualities of latent representations, ARS has wide applicability. First, ARS is valid for various data types as training samples, not limited to images. Second, ARS is task-independent. Shared data representations are reusable for various tasks. Third, prior joint learning frameworks are commonly designed under scenarios of horizontal data partitioning (in which datasets of users share the same feature space but differ in sample ID space), whereas ARS can easily extended its framework to the vertical data partitioning scenario (in which datasets of different users share the same sample Ids but differ in feature columns) \cite{yang2019federated}.

Based on the collaborative learning framework, we apply adversarial example noise \cite{goodfellow2014explaining} to protect shared representations from model inversion attacks. The intuition is that adding special-designed small perturbations on shared data representations can confuse the adversaries so that they cannot reconstruct the original data or particular private attributes from the obfuscated latent representations. By simulating the behavior of attackers, we generate adversarial noise for potential inverse models. The noise is added to data representations before sharing them, in order to make it hard to recover the original inputs. In the meantime, the scale of these perturbations are too small to influence data utility. We propose defense strategies against reconstruction attacks and attribute extraction attacks respectively.

The main contributions of our work are summarized as follows:
\begin{itemize}
	\item We propose ARS, a new paradigm for collaborative framework which is based on representation learning. Different from MPC and FL, ARS is decentralized and has wide applicability.
	
	\item We introduce adversarial noise to defend against model inversion attacks. To the best of our knowledge, we are the first to apply adversarial examples to ensure privacy in collaborative learning.
	
	\item We evaluate our mechanism on multiple datasets and aspects. The results demonstrate that ARS achieves a balance between privacy and utility. We further discuss the limitations and challenges of our work.
\end{itemize}

The remainder of the paper is organized as follows. We first review related work in Section \ref{sect:related_work}. Then we introduce how ARS achieves collaborative learning from an overall perspective in Section \ref{sect:overview}, and detail how adversarial noise is applied to defend against model inversion attacks and ensure privacy in Section \ref{sect:adversarial_noise}. Experimental results are shown in Section \ref{sect:experiments}. In Section \ref{sect:discussion}, we discuss the details and challenges of the work. Finally, we conclude the work in Section \ref{sect:conclusion}.

\section{Related Work}
\label{sect:related_work}

\subsection{Privacy-Preserving Representation Learning}

To avoid privacy leakage in collaborative learning, some previous works focus on learning privacy-preserving representations \cite{xiao2019adversarial,Ferdowsi2020PrivacyPreservingIS}. Latent representations retain the abstract features of data, which can be used for further analysis like classification or regression. Generally, the distribution of data representations can be learned by unsupervised latent variable model, such as autoencoders \cite{ng2011sparse}. Meanwhile, the original information of data can not be directly inferred from representations. For example, in natural language processing, words are transformed into vectors by embedding networks, which is called word2vec \cite{church2017word2vec}. Without embedding networks, it is difficult to recover original words from embedding vectors.

However, data representations are still vulnerable to model inversion attacks \cite{he2019model}. Adversaries can build reconstruction networks to recover original data or reveal some attributes of data from shared representations, even though they have no knowledge of the structure or parameters of the feature extraction models \cite{mahendran2015understanding,he2019model}. For example, they can recover face samples, or infer gender, age and other personal information from shared representations of face images, which were only supposed to be used for training face recognition models.

In order to defend against inversion attacks, recent works focus on adversarial training \cite{xiao2019adversarial} or generative adversarial networks (GANs) \cite{hitaj2017deep}. Attackers' behaviors are simulated by another neural network while learning privacy representations of data, and the two networks compete against each other to improve the robustness of representations. However, since it's hard to achieve a balance between the attacker models and defender models during training \cite{salimans2016improved}, these methods may cost much time in the pretreating phase. Ferdowsi et al. \cite{Ferdowsi2020PrivacyPreservingIS} generate privacy representations by producing sparse codemaps. The above defense methods could be applied to collaborative learning, but the authors didn't give a further discussion on this scenario. In addition, all these methods are task-oriented. When the task of shared data changes, they need to generate new task-oriented data representation again.

\subsection{Adversarial Examples}

Adversarial examples are perturbed inputs designed to fool machine learning models \cite{goodfellow2014explaining}. Formally, we denote by $f: \mathcal{X} \rightarrow \{1, \dots, n\}$ a classifier. For an input $x \in \mathcal{X}$ and a label $l = f(x)$, we call a vector $r$ an adversarial noise if it satisfies:
\begin{equation*}
	\|r\| \leq \epsilon, f(x+r) \neq l,
\end{equation*}
where $\epsilon$ is a small hyper-parameter to adjust the scale of noise.

Adversarial examples have strong transferability. Some works \cite{liu2017delving} have shown that adversarial examples generated for a model can often confuse another model. This property is used to execute transferability based attacks \cite{papernot2017practical}. Even if an attacker has no knowledge about the details of a target model, it can still craft adversarial examples successfully by attacking a substitute model. Therefore, adversarial examples have become a significant threat to machine learning models \cite{goodfellow2014explaining,samangouei2018defense}.

Except for treating adversarial examples as threats, some works utilize the properties of adversarial examples to protect user's privacy \cite{sharif2016accessorize}. In this work, we also use adversarial noise to defend against machine learning based inferring attacks. To the best of our knowledge, we are the first to apply adversarial examples to data sharing mechanisms for collaborative learning.

\begin{figure*}[htbp]
	\centering
	\includegraphics[width=0.9\linewidth]{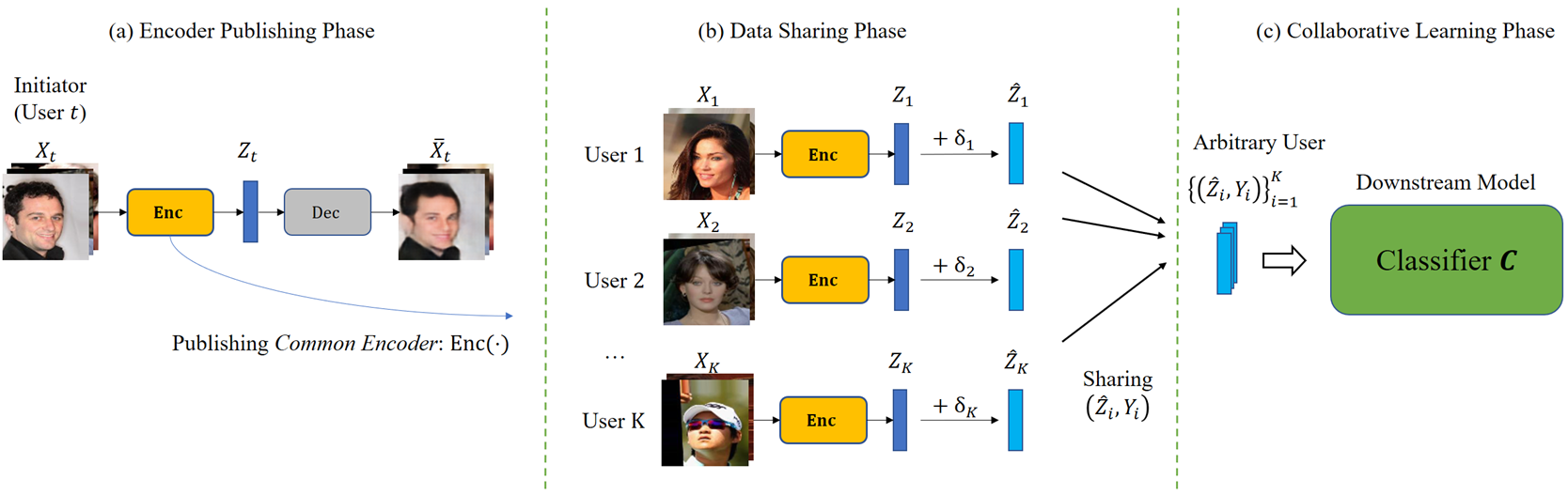}
	\caption{An overview of our basic encoding-based privacy-preserving data sharing mechanism. }
	\label{figure:basic_mechanism}
\end{figure*}

\section{Overview of ARS framework}
\label{sect:overview}

ARS achieves collaborative learning by helping users share data representations instead of raw data to train models. In this section, we first present a joint learning scenario, and propose standards for evaluating the effectiveness of a framework. Then we introduce how ARS works in the whole collaborative learning process. Specifically, how participants encoder their data into latent representations, how to share data representations with others, and how shared data is used for various machine learning tasks.

\subsection{Collaborative Learning Scenario}
\label{sect:data_sharing_scenario}

Consider the scenario where $K$ parties share local data for collaborative training. Note that data belonging to distributed parties can be partitioned horizontally (which means datasets shares same feature space but differ in sample ID space) or vertically (which means datasets shares different same feature space while sample ids are same). For simplicity, we assume data is partitioned horizontally (extension to vertical data partitioning in Section \ref{sect:extension}). The dataset of the $i$-th party is represented as $\{(X_i,Y_i)\} = \{(x_i^1,y_i^1),$ $(x_i^2,y_i^2), \dots, (x_i^{N_i},y_i^{N_i})\}$, where $(x_i^j,y_i^j)$ is a pair of training sample and corresponding label, and $N_i$ is the number of samples. The goal of each participant is to encode their data $X_i$ into representations $Z_i = \{z_i^1, z_i^2, \dots, z_i^{N_i}\}$, and finally train deep learning models on ${Z_1, Z_2, \dots, Z_K}$ which is published by all the data owners.

To privacy concerns, participants have to defend against model inversion attacks that reconstruct original inputs from latent representations. Furthermore, there might be some sensitive features or attributes in training samples. Any label of a training sample can become the machine learning object or a private attribute, depending on the particular user and tasks. For example, when building a disease detection model, the illness of each patient can be a public label, while the identity information such as gender, age should be private. We denote by $M$ the number of private attributes that the $i$-th party has, and denote the set of labels of $k$-th private attributes by $A_{ik} = \{a_{ik}^1, a_{ik}^2, \dots, a_{ik}^{N_i}\}$, where $k \in \{1, 2, \dots, M\}$. Attackers can also conduct attribute extraction attacks to recover these private attributes.

We highlight the following notations which are important to our discussion.

\begin{itemize}
	\item $(X_i, Y_i)$: original training set of user $i$, and label set depending on deep learning tasks.
	
	\item $(x_i^j, y_i^j)$: the $j$-th sample and corresponding label of user $i$.
	
	\item $a_{ik}^j$: value (label) of the $k$-th private attribute of $x_i^j$.
	
	\item $\mathrm{Enc}(\cdot)$: feature extractor which encode inputs into latent representations.
	
	\item $\theta_{\mathrm{Enc}}$: parameters of the model $\mathrm{Enc}$.
	
	\item $z_i^j$: latent representation of the $j$-th sample and label of user $i$, which satisfies $z_i^j=\mathrm{Enc}(x_i^j)$.
	
	\item $C(\cdot)$: downstream deep learning model which takes $\{(Z_i,Y_i)\}_{i=1}^K$ as the training set.
	
	\item $\mathrm{Dec}(\cdot)$: (theoretical) inverse model of $\mathrm{Enc}$, which maps representations to reconstruction of inputs.
	
	\item $\mathrm{SDec}(\cdot)$: decoder trained by adversaries to conduct model inversion attacks, aiming to reconstruct inputs from shared data representations.
\end{itemize}

\subsection{Quantitative Criteria of ARS Mechanism}
\label{sect:effectiveness}

Formally, given a sample set $X$ and a label set $Y$, the main goal of data sharing in ARS is to design a mapping $\mathrm{Enc}: \mathcal{X}\longrightarrow \mathcal{Z}$ that has "nice" properties. These properties consists of utility and privacy of shared data representations as follows. 

\subsubsection{Utility}

After obtaining the shared data, each user can use dataset $\{(Z_1,Y_1),(Z_2,Y_2),$ $\dots, (Z_K,Y_K)\}$ to train downstream deep learning model $C: \mathcal{Z}\longrightarrow \mathcal{Y}$, such as classifiers to predict the label $y$ from $z$. Data utility requires the accuracy to approach the results of models trained directly on $\{(X_i, Y_i)\}_{i=1}^K$. Therefore, $\mathrm{Enc}$ and $C$ should be optimized by minimizing the following expectation:
\begin{equation}
	\mathbb{E}_{x,y \sim p(xy)} \left ( \left | \mathbb{E}_{z\sim p_{\mathrm{Enc}}(z|x)} \left [p_C(y|z) \right ] - p(y|x) \right | \right ),
	\label{equ:utility}
\end{equation}
which ensures that encoding data into representations will not lose much data utility.

\subsubsection{Privacy}

Privacy characterizes the difficulty of finding a model $\mathrm{Dec}: \mathcal{Z} \longrightarrow \mathcal{X}$ to recover raw data, such as visualization information of  inputs of picture type.  Without loss of generality, we focus on reconstruction attacks (attribute extraction attacks is discussed in Section \ref{sect:extension}). Therefore, ARS minimizes privacy leakage $\mathcal{L}$, which can be defined as reconstruction loss:
\begin{equation}
	\mathcal{L} = \mathcal{L}_R = -\frac{1}{N}\sum_{j=1}^{N} \mathcal{D}\left ( x^j, \mathrm{Dec}(\mathrm{Enc}(x^j)) \right ),
	\label{equ:LR}
\end{equation}
where distance function $\mathcal{D}(\cdot, \cdot)$ is used to describe the similarity of original and reconstructed samples. In practice, the distance function is often defined as $l_p$ norm. In computer vision, distance between two images can also be measured as PSNR or SSIM \cite{hore2010image}.

If we consider attribute extraction attacks, we similarly use feature loss to indicate how possible an attacker can predict private features successfully. Note that we don't adopt the error between the true value of private attributes and prediction results of attackers, in case adversaries break the defense by just flipping their results. Instead, we calculate the distance between the prediction result and a fixed vector. The feature loss corresponding to the $k$-th privacy attribute is defined as:
\begin{equation}
	\mathcal{L}_{A_k} = \frac{1}{N}\sum_{j=1}^{N} \mathcal{D}\left ( r_k, F_k(\mathrm{Enc}(x^j)) \right ),
\end{equation}
where $r_k$ is a fixed vector generated randomly, whose size is the same as $a_{ik}$. $F_k$ is a corresponding attribute extraction network trained by adversaries. Low $\mathcal{L}_{A_k}$ implies meaningless prediction results. In this condition, the encoding system minimizes the overall generalization loss:
\begin{equation}
	\mathcal{L} = \lambda_0 \mathcal{L}_R +\sum_{k=1}^M \lambda_k \mathcal{L}_{A_k},
\end{equation}
where $\sum_{k=0}^M \lambda_k = 1$. The value of $\lambda_k$ is assigned depending on users' privacy requirements.

\subsection{The ARS Collaborative Learning Framework}
\label{sect:framework}

We apply autoencoder \cite{ng2011sparse}, an unsupervised neural network to learn data representations. Autoencoder can be divided into an encoder part and a decoder part. The encoder transforms inputs into latent representations, while the decoder map representations to reconstruction of inputs with size smaller than inputs. The optimization target of the autoencoder is to minimize the difference between original inputs and reconstructed ones. Since the encoder compress the feature space, latent representations are wished to remain high-level features of input data \cite{erhan2010does}.

As shown in Fig. \ref{figure:basic_mechanism}, the ARS collaborative learning framework consists of three phases:
\begin{itemize}
	\item \textit{Encoder publishing phase}. In this phase, a common encoder $\mathrm{Enc}$ is trained and published to all participants. The user having permission to train $\mathrm{Enc}$ is called \textit{initiator}, who can be selected randomly, or the party owning the most amount of data. The mechanism in which users shares the same encoder is to ensure that each column of representation vectors shared by any user expresses a specific meaning. Suppose that user $t$ is chosen, it first trains an autoencoder on its local dataset and then publish the encoder part ($\mathrm{Enc}$), as:
	\begin{equation}
		\begin{aligned}
			&\theta_{\mathrm{Enc}}, \theta_{\mathrm{Dec}} = \mathop{\arg \min}_{\theta_{\mathrm{Enc}}, \theta_{\mathrm{Dec}}}\ \frac{1}{N_t}\sum_{j=1}^{N_t} \mathcal{D} \left ( x_t^j, \bar{x}_t^j \right ), \\
			&\mbox{where} \quad \bar{x}_t^j = \mathrm{Dec}(\mathrm{Enc}(x_t^j)).
		\end{aligned}
	\end{equation}
	
	\item \textit{Data sharing phase}. After the common encoder $\mathrm{Enc}$ is published, user $i$ encode its data into representations as: $z_i^j = \mathrm{Enc}(x_i^j)$. Then it shares the pair $\{(Z_i,Y_i)\}$ with the other parties. (In fact, users generate noise $\delta$ for each data representation, then publish $\hat{z}=z+\delta$ instead of $z$ to defend against model inversion attacks. See Section \ref{sect:adversarial_noise}).
	
	\item \textit{Collaborative learning phase}. Finally, each participant can use the pairs $\{(Z_i,Y_i)\}_{i=1}^K$ to train deep learning models locally. For example, model $C: \mathcal{Z}\longrightarrow \mathcal{Y}$ can be trained by minimize the \textit{empirical risk}:
	\begin{equation}
		\hat{R}(C) = \frac{1}{N}\sum_{j=1}^{N} \mathcal{L}(y^j, C(z^j)),
	\end{equation}
	where $N$ is the amount of $(z,y)$ pairs shared by all users, and $\mathcal{L}$ is the loss function of model $C$.
	
\end{itemize}

The ARS framework provides data utility from two aspects. First, accuracy of the deep learning task (i.e. prediction of label $y$) is ensured. Equation (\ref{equ:utility}) can be explained as: for any $x$ sampled from the input distribution and its ground truth label $y$, we aim to maximize $\mathbb{E}_{z\sim p_{\mathrm{Enc}}(z|x)} \left [p_C(y|z) \right ]$ to make it close to $p(y|x)$ (which should be slightly less than $1$). The goal is equivalent to minimize the \textit{generalization error} of $C$: 
\begin{equation}
	R(C) = \mathbb{E}_{x,y \sim p(xy)} \left [\mathbb{E}_{z\sim p_{\mathrm{Enc}}(z|x)} [ \mathcal{L}(y, C(z))] \right ].
\end{equation}
Empirically, models with less empirical risk (training error) normally has less generalization error as well. ARS reduces $R(C)$ by decreasing $\hat{R}(C)$ in the collaborative learning phase. If dataset $\{X_i, Y_i\}_{i=1}^K$ is representative and $\mathrm{Enc}$ is well trained, the joint learning framework will not lose much accuracy.

Second, shared data representations are adaptive to various tasks. Since $\mathrm{Enc}$ is trained without any supervision or knowledge about collaborative learning tasks, we have $p(z|x)=p(z|xy)$ indicating that latent representations are task-independent. Meanwhile, data representations are valid for downstream tasks, which can be explained by a basic idea that learning about the input distribution helps with learning about the mapping from inputs to outputs \cite{goodfellow2016deep}. A well trained autoencoder remains high level features of inputs, which are also useful for supervised learning tasks.

\begin{figure*}[htbp]
	\begin{minipage}[t]{0.49\linewidth}
		\centering
		\includegraphics[width=\linewidth]{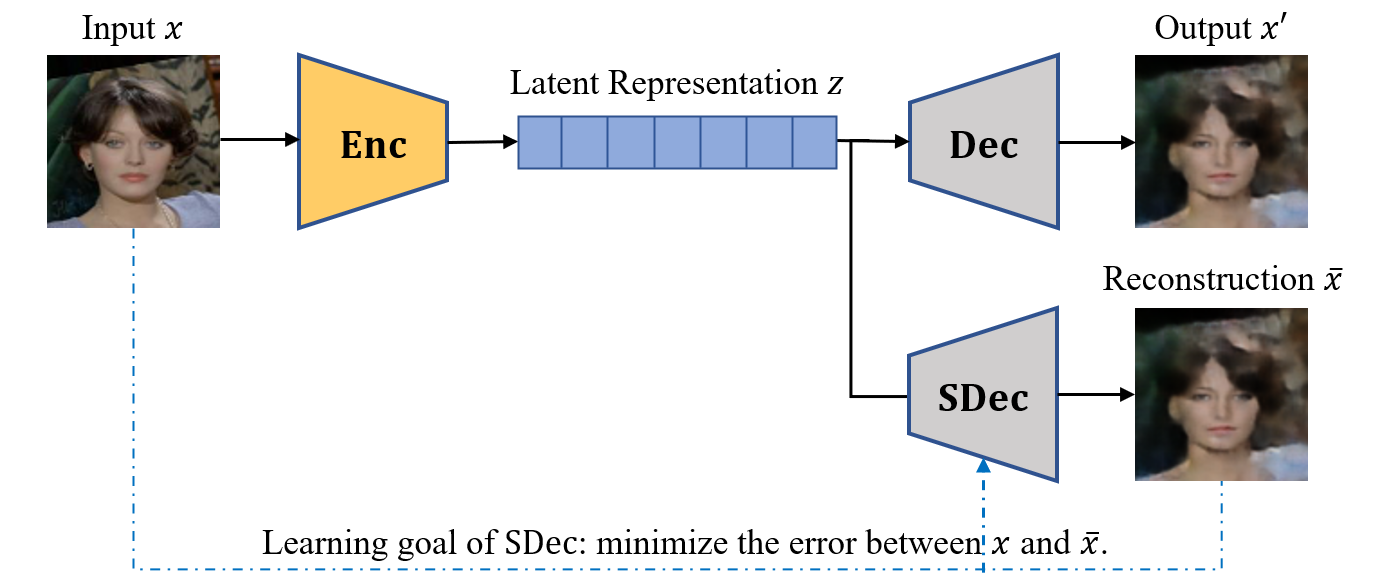}
		\caption{Model inversion attack to reconstruct samples.}
		\label{figure:SDec}
	\end{minipage}
	\begin{minipage}[t]{0.49\linewidth}
		\centering
		\includegraphics[width=\linewidth]{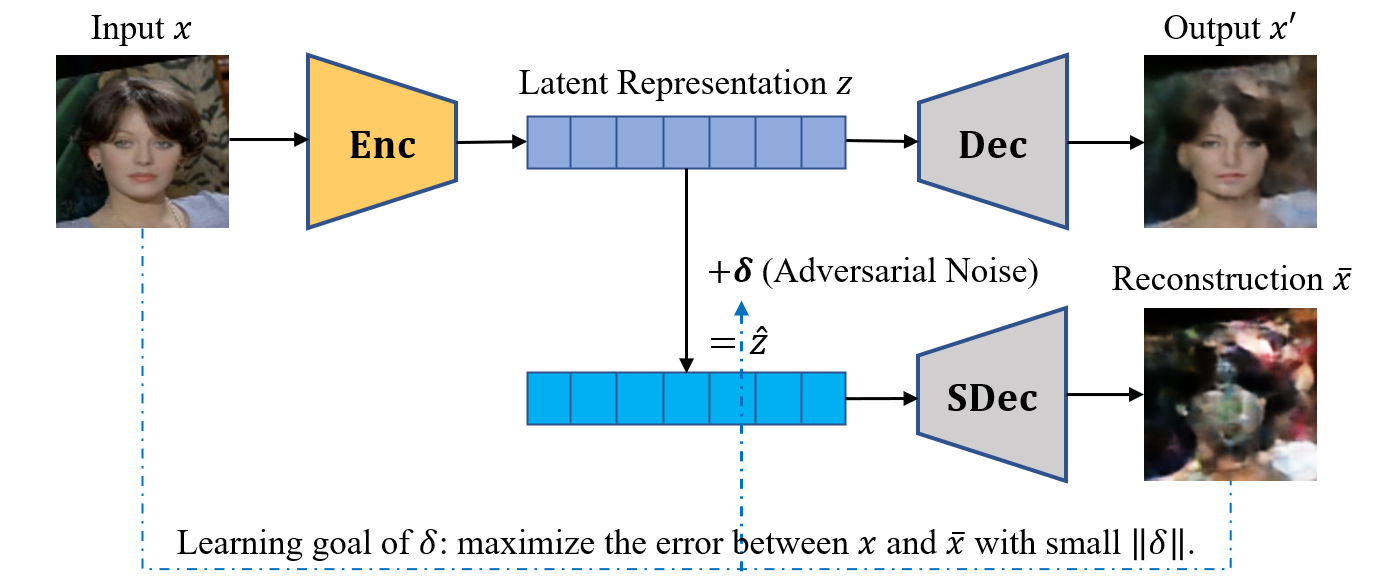}
		\caption{Adding adversarial noise on latent representations.}
		\label{figure:add_adversarial_noise}
	\end{minipage}
\end{figure*}

\section{Adversarial Noise Mechanism Against Inversion Attacks}
\label{sect:adversarial_noise}

This section discusses the privacy of ARS. We first define a threat model caused by model inversion attacks. Then we describe how adversarial noise protects data representations from the inversion attacks. Finally, we extend ARS to the scenarios of vertical data partitioning and attribute extraction attacks.

\subsection{Threat Model}
\label{sect:threat_model}
The main threat of ARS collaborative learning is model inversion attacks. Although private information of raw data is hidden into representations, it may still be recovered by reconstruction models, like $\mathrm{Dec}$ trained by initiator in the encoder publishing phase. For simplicity, we firstly focus on reconstruction attacks, and extend our method to overcome attribute extraction attacks in Section \ref{sect:extension}.

The threat model defines adversaries who act like curious participants and try to recover the original input $x$ from data representation $z=\mathrm{Enc}(x)$ for private information. Like other participants, adversaries can obtain data representations shared by users and have query access to the published encoder $\mathrm{Enc}$, but have no knowledge about the architecture and parameters of $\mathrm{Enc}$. Nonetheless, adversaries can still exploit query feedback of $\mathrm{Enc}$ to carry out black-box attacks and acquire substitute decoders of $\mathrm{Dec}$. For instance, if user $i$ with data $\{(X_i,Y_i)\}$ is an attacker, it can generate latent representations by querying $z_i^j = \mathrm{Enc}(x_i^j)$ for $N_i$ times. Then it can build a substitute decoder $\mathrm{SDec}: \mathcal{Z}\longrightarrow \mathcal{X}$ as Figure \ref{figure:SDec} shows. $\mathrm{SDec}$ can be optimized by minimizing the distance between original and recovered samples:
\begin{equation}
	\theta_{\mathrm{SDec}} = \mathop{\arg \min}_{\theta_{\mathrm{SDec}}} \frac{1}{N_i}\sum_{j=1}^{N_i} \mathcal{D} \left ( x_i^j, \mathrm{SDec}(\mathrm{Enc}(x_i^j)) \right ).
	\label{equ:equ2}
\end{equation}

This kind of attack can be regarded as \textit{chosen-plaintext attack} (CPA) from the cryptographic point of view. The reconstruction ability of model inversion attacks is strong when training samples of victims and attackers has similar distribution. In the following, we propose the adversarial noise mechanism to defend against the inversion attacks.

\subsection{Adding Adversarial Noise}
\label{sect:simple_adversarial_noise}

The strategy to defense against inversion attacks comes from a simple idea: \textit{adding an intentionally designed small noise on latent representations before sharing them}. On the one hand, the scale of the noise vector is so small that it would not reduce the utility of shared data representations. On the other hand, we hope adding noise on data representations can make data reconstructed by $\mathrm{SDec}$ different enough from the original inputs. Inspired by adversarial examples, we let users add adversarial noise to latent representations in set $Z$, and share the set of adversarial examples $\hat{Z}$ instead, as shown in Figure \ref{figure:add_adversarial_noise}.

According to some researches \cite{papernot2017practical}, adversarial examples have strong transferability. Empirically, if $\hat{Z}$ successfully fools the decoder from which it is generated, then it is likely to cause another decoder with similar decision boundary to recover samples that are very different from the original inputs. The properties of adversarial noise determine its ability to protect the privacy of shared data representations even if the scale of noise is small. 

The method to generate effective adversarial noise consists of two steps. A participant first trains a substitute decoder $\mathrm{SDec}$ locally by simulating reconstruction attacks; then it generates adversarial noise for $Z$ to make $\mathrm{SDec}$ invalid. Adversarial noise is generated through iterative fast gradient sign method (I-FGSM) \cite{goodfellow2014explaining}, which sets the direction of adversarial noise to the gradient of objective function $\mathcal{L}_R$ with respect to $z$. The adversarial latent representation $\hat{z}$ is calculated as:
\begin{equation}
	\begin{aligned} 
		& \hat{z}^{(t+1)} = \hat{z}^{(t)} + \alpha \cdot \mathrm{sign} \left(\nabla_z \mathcal{D} \left ( x, \mathrm{SDec}(\hat{z}^{(t)}) \right ) \right ), \\
		& t=1,\dots ,n, \quad \hat{z}^{(1)}=z, \quad \hat{z} = \hat{z}^{(n+1)}, 
	\end{aligned}
	\label{equ:equ3}
\end{equation}
where $\alpha$ is a hyper-parameter regulating the scale of noise in each iteration and $n$ is the iteration time. The adversarial noise on $z$ can be denoted as:
\begin{equation}
	\delta_{R}=\hat{z}-z.
\end{equation}

In consideration of data utility, the difference between $z$ and $\hat{z}$ should not be so great, otherwise $\hat{z}$ would lose most of the features of $x$. Therefore, given an encoded vector $z$, we must ensure that $|\delta_{R}| = |\hat{z}-z| \leq \epsilon$, where $\epsilon$ is a hyper-parameter to be chosen. Next we prove that adversarial noise in Equation (\ref{equ:equ3}) satisfies the privacy budget.

\begin{prop}
	Given $\epsilon \in \mathbb{R}^{*}$, $n \in \mathbb{N}^{*}$. Suppose $\alpha = \frac{\epsilon}{n}$, then adversarial data representation $\hat{z}$ defined by Equation (\ref{equ:equ3}) satisfies: $|\hat{z}-z| \leq \epsilon$.
	\label{prop:budget}
\end{prop}

\begin{proof}
	For any iteration step $t \in \{1,\dots ,n\}$, we have:
	\begin{equation*}
		\begin{aligned}
			\left |\hat{z}^{(t+1)}-\hat{z}^{(t)} \right | & = \left | \alpha \cdot \mathrm{sign} \left(\nabla_z \mathcal{D} \left ( x, \mathrm{SDec}(\hat{z}^{(t)}) \right ) \right ) \right | \\
			& = \alpha \left | \mathrm{sign} \left(\nabla_z \mathcal{D} \left ( x, \mathrm{SDec}(\hat{z}^{(t)}) \right ) \right ) \right | \\
			& = \alpha.
		\end{aligned}
	\end{equation*}
	Therefore, $|\delta_{R}| = |\hat{z}-z| = |\hat{z}^{(n+1)}-\hat{z}^{(1)}| \leq \sum_{t=1}^n |\hat{z}^{(t+1)}-\hat{z}^{(t)}| = n \cdot \alpha = \epsilon$.
\end{proof}

Proposition \ref{prop:budget} shows that if we set $\alpha$ to $\frac{\epsilon}{n}$ in Equation (\ref{equ:equ3}), then the scale of adversarial noise will be limited to $\epsilon$. Here $\epsilon$ is called \textit{defense intensity}, which determines the utility and privacy of representations.

\subsection{Adding Masked Adversarial Noise}

$\hat{Z}$ can mislead most of decoders trained from pairs $\{(Z, X)\}$ due to the transferability of adversarial examples. However, model inversion can still occur if attackers apply \textit{adversarial training} \cite{goodfellow2014explaining} (also called data enhancement) to build $\mathrm{SDec}': \mathcal{\hat{Z}}\longrightarrow \mathcal{X}$. Suppose an attacker (user $i$) aims to execute a reconstruction attack. In the data sharing phase, the attacker first trains $\mathrm{SDec}$ on its local data $X_i$ by optimizing Equation (\ref{equ:equ2}). Then it transforms $X_i$ into representations $Z_i$ and adding adversarial noise to them as Equation (\ref{equ:equ3}) shows. After that, the attacker can train $\mathrm{SDec}'$ on $\{(\hat{Z}_i, X_i)\}$ as:
\begin{equation}
	\theta_{\mathrm{SDec}'}=\mathop{\arg \min}_{\theta_{\mathrm{SDec}'}}\frac{1}{N_i}\sum_{j=1}^{N_i} \mathcal{D} \left ( x_i^j, \mathrm{SDec}'(\hat{z}_i^j) \right ).
	\label{equ:adv_attack}
\end{equation}

To defend against model inversion attack with data enhancement, we propose \textit{noise masking}, a simple and effective method to make participants perturb data representations in unique ways. Each user possess a mask vector (denoted by $\mathbf{m}$) with the same size as data representations. All dimensions of $\mathbf{m}$ are randomly initialized to either $0$ or $1$, in order to mask some dimensions of calculated gradients and thus allow users to perturb other dimensions of representations. The process of generating masked adversarial noise on representation $z$ is expressed as:
\begin{equation}
	\begin{aligned}
		& \hat{z}^{(t+1)} = \hat{z}^{(t)} + \frac{\epsilon}{n} \cdot \mathrm{sign} \left (\mathbf{m} \cdot \nabla_z \mathcal{D} \left ( x, \mathrm{SDec}(\hat{z}^{(t)}) \right ) \right ), \\
		& t=1,\dots ,n, \quad \hat{z}^{(1)}=z, \quad \hat{z} = \hat{z}^{(n+1)}. 
	\end{aligned}
	\label{equ:equ6}
\end{equation}

The value of a vector $\mathbf{m}$ is held by its owner secretly, which leads attackers to train substitute models on data representations that are perturbed in a significant different way from representations of target users with high probability. We discuss whether and when the masks are effective in Section \ref{sect:noise_masking}, and prove that mask vectors with sufficiently large dimension are difficult to crack through brute force. Therefore, ARS with noise masking can be considered safe enough. Algorithm \ref{alg:alg1} shows the overall process of generating data representations in ARS.

\begin{algorithm}[htbp]
	\caption{ARS Data Sharing Method for User $i$}
	\label{alg:alg1}
	\KwIn{training samples $X_i = \{x_i^1, x_i^2, \dots, x_i^{N_i}\}$}
	\KwOut{representations $\hat{Z}_i = $ $\{\hat{z}_i^1, \hat{z}_i^2, \dots, \hat{z}_i^{N_i}\}$}
	
	initialize $\mathrm{Enc}$, $\epsilon$, $n$, $\mathbf{m}_i$\;
	
	$Z_i=\{z_i^1, z_i^2, \dots, z_i^{N_i}\}$, where $z_i^j = \mathrm{Enc}(x_i^j)$\;
	
	update $\mathrm{SDec}$ via: $\theta_{\mathrm{SDec}} = \mathop{\arg \min}_{\theta_{\mathrm{SDec}}} \frac{1}{N_i}\sum_{j=1}^{N_i} \mathcal{D} \left (x_i^j , \mathrm{SDec}(z_i^j) \right )$\;
	
	\For{$j=1$ to $N_i$}{
		$\hat{z}_i^j = z_i^j$\;
		\For{$k=1$ to $n$}{
			$\hat{z}_i^j = \hat{z}_i^j + \frac{\epsilon}{n} \cdot \mathrm{sign} \left (\mathbf{m}_i \cdot \nabla_z \mathcal{D} \left ( x, \mathrm{SDec}(\hat{z}_i^j) \right ) \right )$\;
		}
	}
	return $\hat{Z}_i = $ $\{\hat{z}_i^1, \hat{z}_i^2, \dots, \hat{z}_i^{N_i}\}$\;
\end{algorithm}

\subsection{Extension}
\label{sect:extension}

\subsubsection{Attribute Exaction Attacks}
In addition to reconstruction attacks, user $i$ acting as an adversary can also train a classifier $F_k: \mathcal{Z} \longrightarrow \mathcal{A}_k$ on its local pairs $\{(Z_i,A_{i,k})\}$, as:
\begin{equation}
	\theta_{F_k} = \mathop{\arg \min}_{\theta_{F_k}} \frac{1}{N_i}\sum_{j=1}^{N_i} \mathcal{D} \left ( a_{ik}^j, F_k(z_i^j) \right ),
	\label{equ:equ2_2}
\end{equation}
where $k \in \{1, 2, \dots, M\}$ is the index of the target private attribute. After that, the attacker can conduct attribute extraction attacks by predicting the $k$-th private attribute from shared data representations.

The strategy to overcome the feature leakage is similar to the above method. To preserve the $k$-th private attribute, users should first train a classifier $F_k$ locally, and then craft adversarial noise on $Z$ to maximize the feature loss $\mathcal{L}_{A_k}$. Here we simply apply FGSM method as:
\begin{equation}
	\delta_{A_k} = -\epsilon \cdot \mathrm{sign} \left ( \nabla_z \mathcal{D}\left ( r_k, F_k(\mathrm{Enc}(x^j)) \right ) \right ),
	\label{equ:equ3_2}
\end{equation}
where $\delta_{k}$ is the adversarial noise to preserve the $k$-th private attribute. The purpose is to keep prediction results close to a certain vector $r_k$, regardless of inputs, thus lead the prediction to meaningless results.

Consequently, the overall adversarial noise of an arbitrary participant $i$ can be calculated as:
\begin{equation}
	\delta_{i} = \lambda_{0}\delta_{R} + \sum_{k=1}^{M}\lambda_{k}\delta_{A_k},
	\label{equ:equ4_2}
\end{equation}
where $\sum_{k=0}^{M}\lambda_{k} = 1$, so that $|\delta_{i}|\leq \epsilon$. Experimental results in Section \ref{sect:privacy_protection2} show that letting $\lambda_0=\frac{1}{2}$, $\lambda_k=\frac{1}{2M}$ $(k\geq 1)$ may be a good choice to ensure the defense against data leakage and feature leakage at the same time.

\subsubsection{Vertical Data Partitioning}
In the second extension, we generalize ARS to the vertical data partitioning scenario. For simplicity, we suppose that datasets are aligned on IDs, and each participant owns $N$ samples. In the training phase, user $i$ trains $\mathrm{Enc}_i$ on its local dataset $X_i=\{ x_i^1, x_i^2, \dots, x_i^N\}$, then calculate $z_i^j=\mathrm{Enc}_i(x_i^j)$ and share $\hat{z}_i^j = z_i^j+\delta_{i}^j$. One of the users share labels $Y=\{ y^1, y^2, \dots, y^N\}$. Different from horizontal data partition, a common encoder is unused, so $\mathrm{Enc}_i$ does not need to be shared.

To train downstream models collaboratively, data representations shared by each users are concatenated as $w^j = \mathrm{concat}(\hat{z}_1^j,$ $\hat{z}_2^j, \dots, \hat{z}_K^j)$. Afterwards, any user $i$ who has sufficient computational power is able to train model $C_i$ on $\{(W, Y)\}$ locally.

\section{Experiments}
\label{sect:experiments}
In this section, we evaluate ARS by simulating a multi-party collaborative learning scenario. We present the performance of ARS in privacy preserving, as well as compare it with other joint learning frameworks, and then study the effectiveness of our mechanism in protecting private attributes.

\subsection{Experiment Settings}
\subsubsection{Datasets}
The experiments are conducted on three datasets: MNIST \cite{mnist} and CelebA \cite{Liu2014DeepLF} for horizontal data partitioning, and Adult  \cite{Dua2019} from UCI Machine Learning Repository for vertical data partitioning. MNIST consists of 70000 handwritten digits, the size of each image is $28 \times 28$. CelebA is a face dataset with more than 200K images, each with 40 binary attributes. Each image is resized to $96 \times 96 \times 3$. Adult dataset is census information, and the given task is to determine whether a person makes over \$50,000 a year basing on 13 attributes including "age", "workclass", "education". In our experiments, the inputs are mapped into real vectors with 133 feature columns.

\subsubsection{Scenario}
We simulate horizontal data partitioning scenarios where the number of participants $K=5$ in MNIST, and $K=3$ in CelebA. Each participant is randomly assigned 10000 examples as a local dataset. When conducting experiments on MNIST, the common encoder $\mathrm{Enc}$ and each user's substitute decoder $\mathrm{SDec}$ are implemented by three-layer ReLU-based fully connected neural networks. On CelebA, $\mathrm{Enc}$ and $\mathrm{SDec}$ are implemented by four-layer convolutional neural networks. In the data sharing phase, the iteration time $n$ is set to 10. We suppose that attackers execute adversarial training attacks (data enhancement) for two types of objectives: to recover original samples (see Section \ref{sect:privacy_protection}), or to extract private attributes (see Section \ref{sect:privacy_protection2}). In the collaborative learning phase, the tasks are set as training classifiers on shared data representations $\{\hat{Z}_i\}_{i=1}^K$. The labels are 10-dimensional one-hot codes in MNIST, and 2-dimensional vectors corresponding to binary attributes in CelebA.

We also present a vertical data partitioning scenario on Adult dataset in Section \ref{sect:vertical_experiment}. Each user holds $20000$ aligned samples, while the number of column vectors owned by different users is kept as close as possible. The encoders and substitute decoders are implemented by four-layer ReLU-based fully connected neural networks. An adversary owning $5000$ partial inputs would like to attack the user who possess the same feature columns of samples.

All above tasks can be regarded as case studies of collaborative learning in the real world. For example, companies can share privacy representations of photos to train face recognition models; enterprises may cooperate with each other to draw more comprehensive customer personas.

\subsubsection{Privacy Metrics}
We consider three metrics to measure $\mathcal{D}(x,\bar{x})$, which reflects the privacy leakage caused by reconstruction attacks. According to a common practice that represents $\mathcal{D}(x,\bar{x})$ as $l_p$ norm, we define $\mathcal{D}(x,\bar{x})$ as $\| x-\bar{x} \|_2$. This metric is also well known as MSE (Mean Square Error). PSNR (Peak Signal-to-Noise Ratio) and SSIM (Structural Similarity Index) \cite{hore2010image} reflects pixel level difference between original and reconstructed images, and are highly consistent with human perceptual capability, so they are also applicable for evaluating privacy leakage. The metrics are widely adopted in related researches, which allows us to compare their results with ours directly. To measure the privacy of private attributes, we simulate attacks on these attributes and record the proportion that the predictions are equal to the fixed vector.

\begin{figure}[htbp]
	\centering
	\subfloat[MNIST, Accuray]{
		\includegraphics[width=0.45\linewidth]{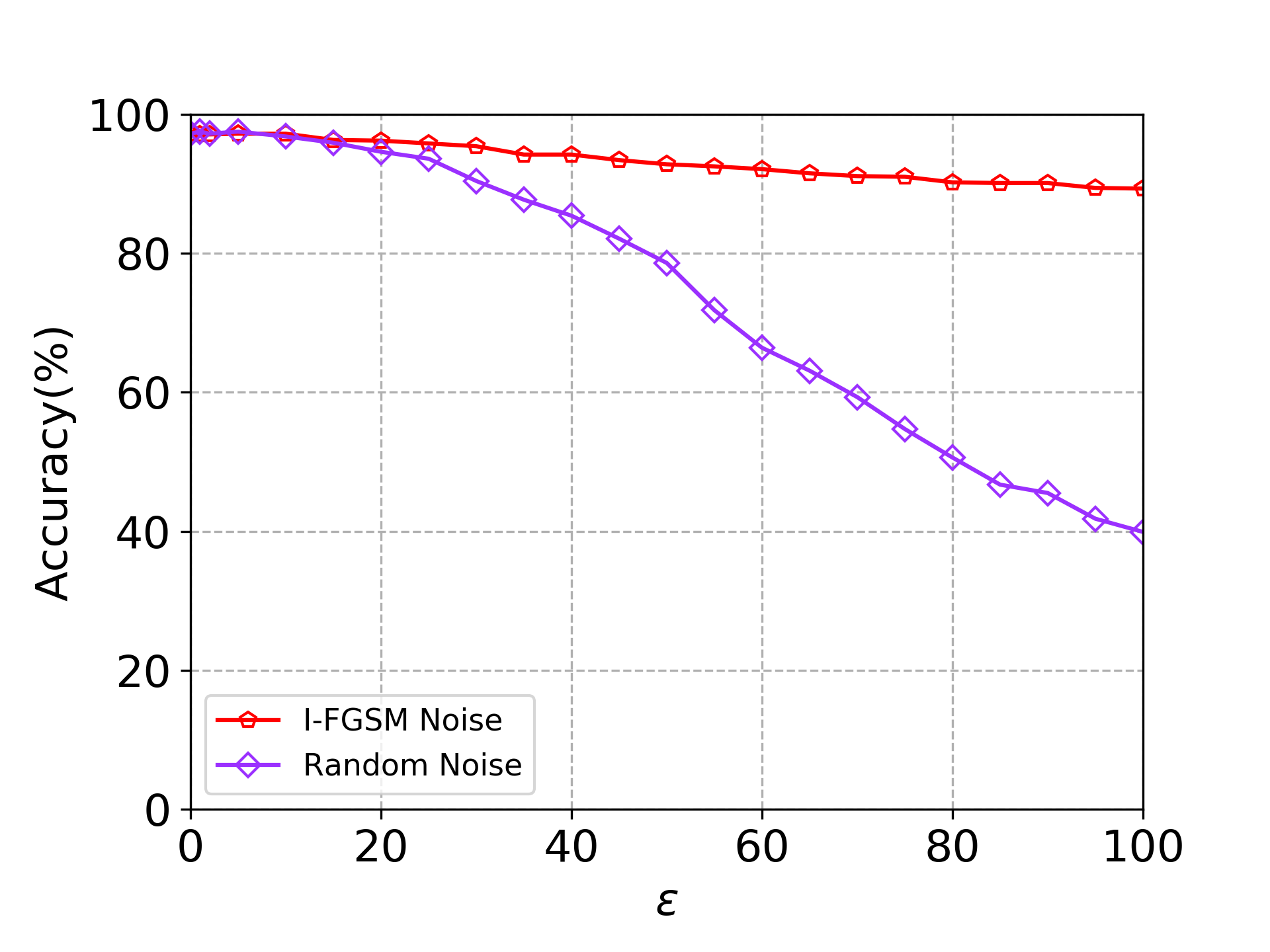}
	}%
	\subfloat[MNIST, MSE]{
		\includegraphics[width=0.45\linewidth]{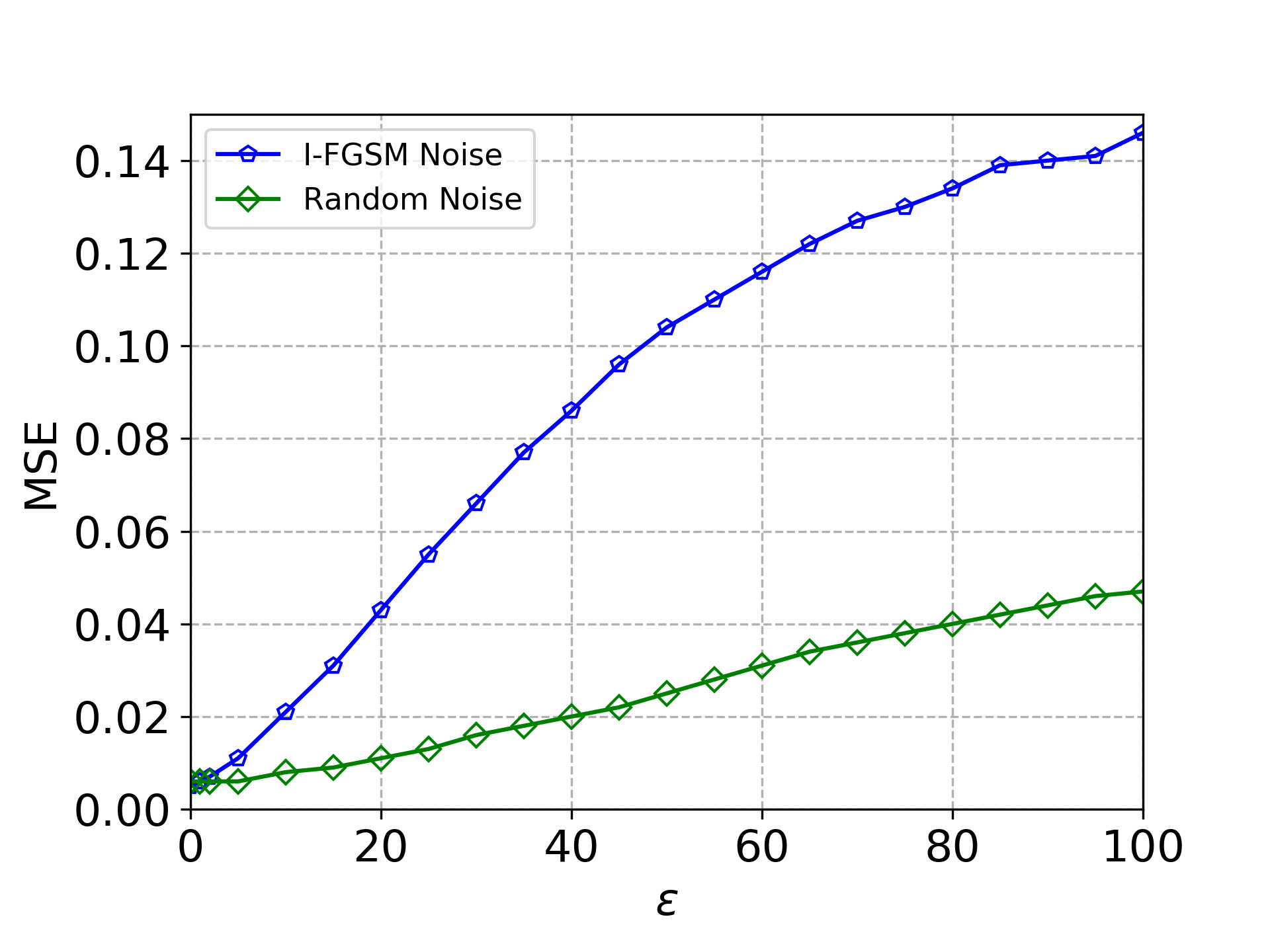}
	}%
	
	\subfloat[CelebA, Accuray]{
		\includegraphics[width=0.45\linewidth]{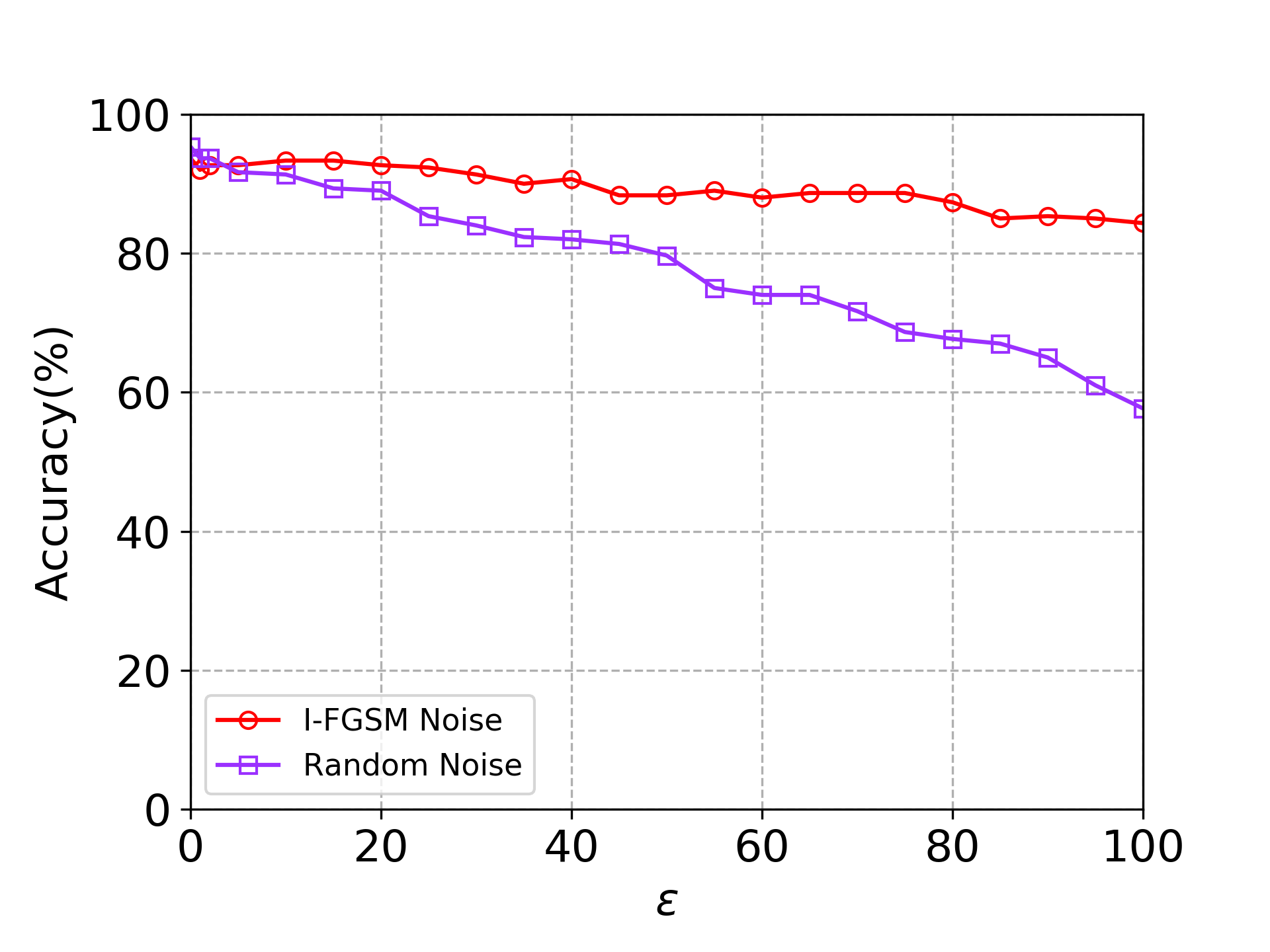}
	}%
	\subfloat[CelebA, MSE]{
		\includegraphics[width=0.45\linewidth]{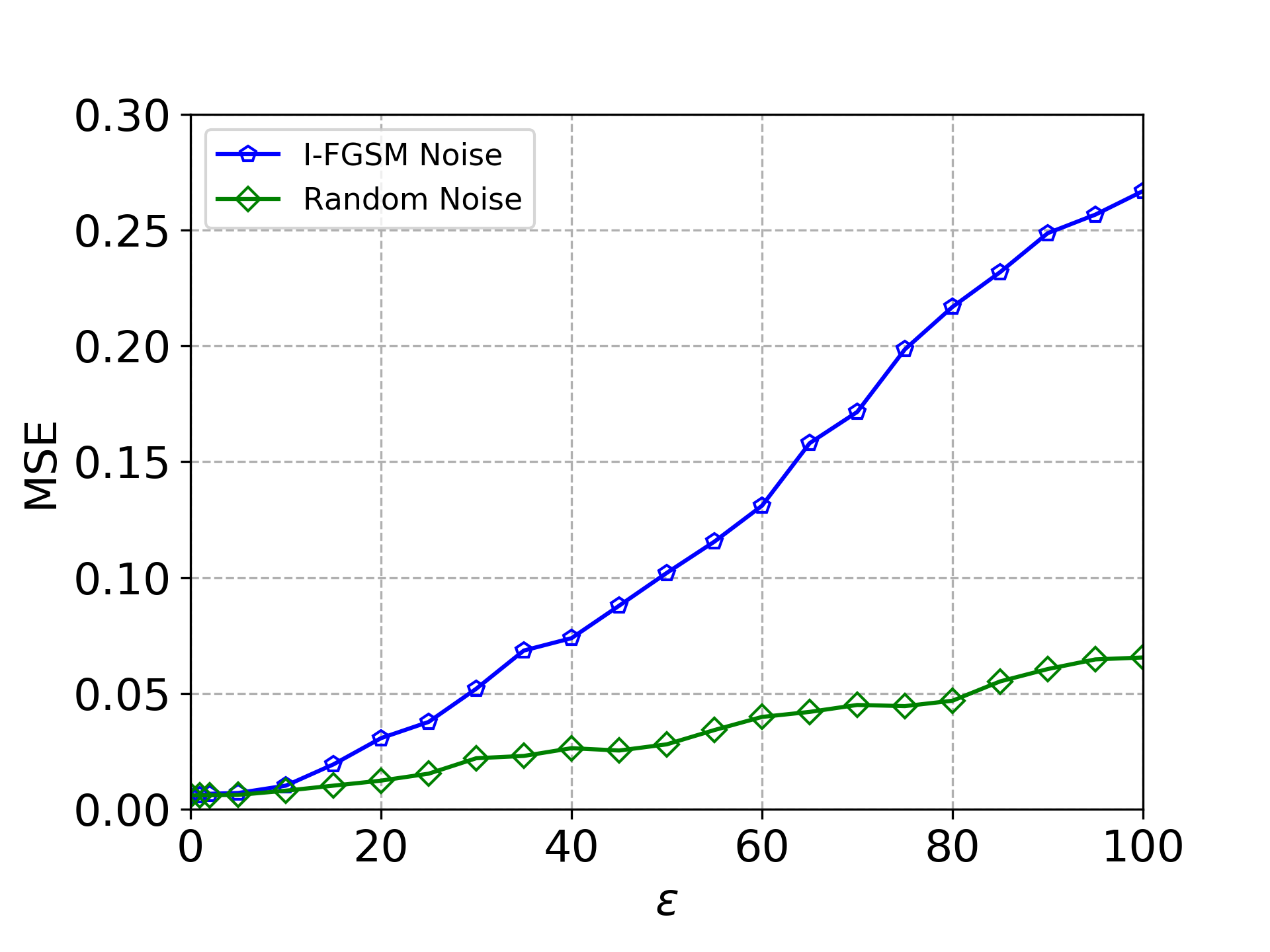}
	}%
	\centering
	\caption{Classification accuracy and reconstruction loss versus $\epsilon$.}
	\label{figure:experiment1}
\end{figure}

\begin{figure}[tbp]
	\centering
	\subfloat[]{
		\includegraphics[width=0.3\linewidth]{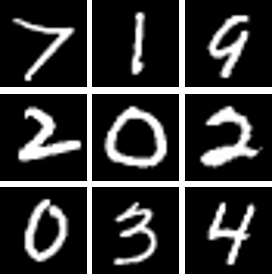}
	}%
	\subfloat[]{
		\includegraphics[width=0.3\linewidth]{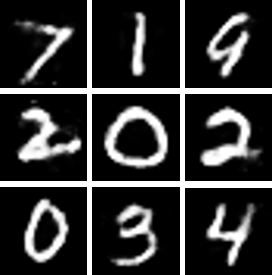}
	}%
	\subfloat[]{
		\includegraphics[width=0.3\linewidth]{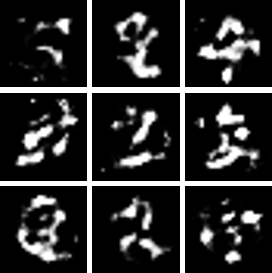}
	}%
	\centering
	\caption{Digit images and corresponding reconstructed images. (a) Original input images. (b) Reconstructed images corresponding to data representations without noise. (c) Reconstructed images corresponding to representations with adversarial noise ($\epsilon=50$).}
	\label{figure:experiment2_1}
\end{figure}

\begin{figure*}[htbp]
	\centering
	\subfloat[Original Images]{
		\includegraphics[width=0.125\linewidth]{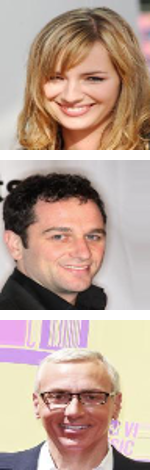}
	}\quad
	\subfloat[$\epsilon=0$]{
		\includegraphics[width=0.125\linewidth]{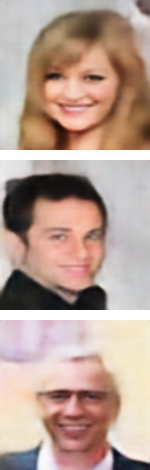}
	}\quad
	\subfloat[$\epsilon=25$]{
		\includegraphics[width=0.125\linewidth]{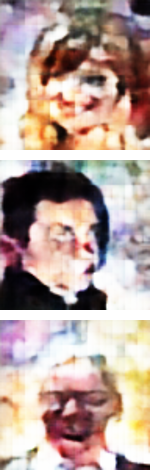}
	}\quad
	\subfloat[$\epsilon=50$, no mask]{
		\includegraphics[width=0.125\linewidth]{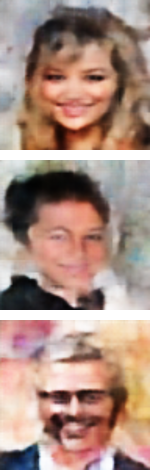}
	}\quad
	\subfloat[$\epsilon=50$]{
		\includegraphics[width=0.125\linewidth]{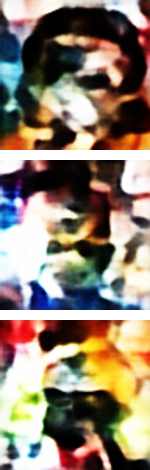}
	}\quad
	\subfloat[$\epsilon=100$]{
		\includegraphics[width=0.125\linewidth]{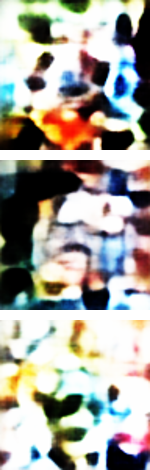}
	}
	\caption{Face images and corresponding reconstructed images. Images in column (a) are raw data. Column (b), (c), (e), (f) corresponds to adversarial latent representations with different $\epsilon$. Column (d) corresponds to adversarial representations ($\epsilon = 50$) without noise masking.}
	\label{figure:experiment2_2}
\end{figure*}

\subsection{Protecting Privacy Against Reconstruction Attacks}
\label{sect:privacy_protection}

\subsubsection{Defense Intensity} We firstly evaluate the utility and privacy of ARS with different defense intensity $\epsilon$. Figure \ref{figure:experiment1} reports the accuracy of the classifier trained on shared data representations to evaluate the utility, and MSE of reconstructed images to evaluate privacy. Experiments are conducted on both MNIST and CelebA datasets. On CelebA, the task is to predict the attribute "Male". We set up another method as a baseline, which generates random noise with uniform distribution instead of adversarial noise. Observe that with the increase of $\epsilon$, the reconstruction loss becomes higher, which indicates that adversarial noise with a larger scale makes it more difficult to filch private information from shared representations. When measuring the utility of shard data, we find that when $\epsilon$ becomes larger, the accuracy decreases slightly from 97.2\% to 89.3\% in MNIST, and from 93.7\% to 84.3\% in CelebA. The variety of MSE and accuracy with the scale of noise illustrates the trade-off between utility and privacy of shared data.

\subsubsection{Visualization of Reconstructed Images}

We then explore the effectiveness of adversarial noise defense by displaying the images under reconstruction attacks. Figure \ref{figure:experiment2_1} compares digit images recovered from adversarial latent representations with the undefended version, and illustrates that $\epsilon=50$ can well ensure data privacy. This preliminarily proves the privacy of the adversarial noise mechanism.

We further study the influence of different $\epsilon$ in CelebA. Figure \ref{figure:experiment2_2} shows the reconstructed images corresponding to data representations adding several kinds of noise. If adversarial noise is not used, the reconstructed image restores almost all private information of faces. When $\epsilon$ is set to 50, the recovered faces lose most of the features used to determine identity. When $\epsilon=100$, the reconstructed images become completely unrecognizable. For further discussion, we present the result when $\epsilon=50$, while the adversarial noise is not masked. As shown in the column (d) of Figure \ref{figure:experiment2_2}, the faces do get blurry, but some features with private information are still retained.

The experiments present satisfactory performance of adversarial noise on latent representations. If defense intensity is set to a sufficiently small value, the shared data can maintain high utility and privacy. In a real application, data utility is expected to be as higher as possible while privacy is well preserved. We choose $\epsilon=50$ as a suitable defense intensity in both datasets in the following experiments, because of its high privacy and acceptable classification accuracy, which is 92.8\% in MNIST, and 88.3\% in CelebA.

\renewcommand\arraystretch{1.3}
\begin{table*}[htbp]
	\centering
	\caption{Comparison of ARS and mainstream methods on important metrics.}
	\begin{tabular}{cccc}
		\toprule
		Framework	& Basic Method	& Reported Accuracy on MNIST	& Communication Cost per Server	\\ \midrule
		SecureML	& MPC	& 93.4\%	& $\succeq O(N \cdot K \cdot d + (B+d)\cdot t)$ \\
		DPFL	& Federated Learning	 & $\leq 92\%$ ($K \leq 1000$)	& $O(K \cdot W \cdot t)$ \\
		ARS	& Representation Learning	& 92.8\% ($\epsilon=50$)	& \bm{$O(N \cdot K \cdot h)$} \\
		\bottomrule
	\end{tabular}
	\label{tab:utility_comparison}
\end{table*}

\subsubsection{Comparison with Existing Mechanisms}
\label{sect:comparison}

We first evaluate data utility by comparing ARS with mainstream joint learning frameworks on MNIST. Table \ref{tab:utility_comparison} summarizes the performance of these methods in terms of important metrics such as classification accuracy and communication cost. SecureML \cite{mohassel2017secureml} is a two-party computation (2PC) collaborative learning protocol based on secret sharing and garbled circuits. Differential private federated learning \cite{geyer2017differentially} (which we call DPFL) approximates the aggregation results with a randomized mechanism to protect datasets against differential attacks. As shown in the table, the classification accuracy of ARS reaches a level similar to that of MPC and FL based methods, when we regard $\epsilon=50$ as a compromise between utility and privacy. Note that ARS is designed for general scenarios rather than specific datasets or tasks, and better results can be achieved by using more complex models or fine tuning hyper-parameters. Moreover, ARS has a low communication cost. We denote by $N$ the average number of training samples owned by each party, by $d$ the dimension of the original data, by $B$ the batch size, by $t$ the number of epochs to train models, by $W$ the number of parameters in a model, and by $h$ the dimension of latent representations. Since the iteration times to train a deep learning model using stochastic gradient descent (SGD) method depends on the number of training samples, otherwise the model will be underfitted, the complexity of $E$ should not be lower than $N$, it is easy to prove that ARS has the lowest communication complexity among the three methods.

We next evaluate the privacy of ARS by simulating reconstruction attacks. Since PSNR and SSIM are widely adopted by the latest researches in this field, we also calculate these two metrics as privacy leakage, and compare ARS with two state-of-the-art data sharing mechanisms: generative adversarial training based sharing mechanism \cite{xiao2019adversarial} and SCA based sharing mechanism \cite{Ferdowsi2020PrivacyPreservingIS}. Similar to ARS, both of them learn representations of data. The experiment is conducted on CelebA. Table \ref{tab:experiment_comparison} reports the privacy leakage of the mechanisms. As we can see, ARS performs better than the other two frameworks in PSNR, even when $\epsilon=50$. When $\epsilon$ increases to 100, SSIM of ARS also reaches the best result of the three mechanisms.

\begin{table}[htbp]
	\centering
	\caption{Privacy of different representation based mechanisms.}
	\begin{tabular}{cccccc}
		\toprule
		& \tabincell{c}{ARS\\($\epsilon = 50$)}  & \tabincell{c}{ARS\\($\epsilon = 100$)} & \tabincell{c}{Baseline\\($\epsilon = 50$)}	& \cite{xiao2019adversarial} & \cite{Ferdowsi2020PrivacyPreservingIS} \\ \midrule
		PSNR & \textbf{9.932}                & \textbf{5.748}                  & 15.527                 & 15.445               & 12.31             \\
		SSIM & 0.531                 & \textbf{0.101}                  & 0.728                  & 0.300                & 0.25             \\
		\bottomrule
	\end{tabular}
	\label{tab:experiment_comparison}
\end{table}

\subsubsection{Vertical Data Partitioning}
\label{sect:vertical_experiment}

We next evaluate the performance of ARS on Adult dataset under the vertical data partitioning scenario. In Table \ref{tab:experiment_vertical}, we divide the totally 133 feature columns of samples as evenly as possible among three users, then present the prediction accuracy (Acc), $F_1$-score, and robustness to reconstruction attacks of the collaborative learning system with different numbers of participants. To confirm the effectiveness of the inversion attack, we show the accuracy of adversarial training based reconstruction from data representations without adversarial noise (Adv. Tr.). Here we suppose that encoders used to generate data representations can be obtained by adversaries. To estimate the privacy of ARS, we present the reconstruction accuracy on concatenated adversarial latent representations (Rec. Acc). We observe that large number of collaborating users leads to higher prediction accuracy and $F_1$-score of the downstream classification model, and when the noise scales up, the reconstruction accuracy drops to less than $50\%$. So we demonstrate the great utility and privacy of ARS under the vertical data partitioning scenario.

\begin{table*}[htbp]
	\centering
	\caption{Results with different numbers of data owners on the vertical data partitioning scenario.}
	\begin{tabular}{c|cccc|cccc|cccc}
		\hline
		\multirow{2}{*}{$\epsilon$} & \multicolumn{4}{c|}{K=1}      & \multicolumn{4}{c|}{K=2}      & \multicolumn{4}{c}{K=3}       \\ \cline{2-13} 
		& Acc  & $F_1$-score   & Adv. Tr. & Rec. Acc   & Acc  & $F_1$-score   & Adv. Tr. & Rec. Acc   & Acc  & $F_1$-score   & Adv. Tr. & Rec. Acc   \\ \hline
		0                    & 78.9\% & 87.2\% & 97.4\%     & 97.3\% & 82.2\% & 88.7\% & 97.6\%     & 97.3\% & 84.6\% & 89.9\% & 97.7\%     & 97.5\% \\
		10                   & 78.9\%   & 87.1\% & 97.2\%     & 77.5\% & 81.6\% & 88.1\% & 97.4\%     & 82.8\% & 83.8\% & 89.3\% & 97.2\%     & 84.7\% \\
		25                   & 78.5\%   & 86.8\% & 96.7\%     & 45.4\% & 81.2\% & 87.7\% & 93.4\%     & 61.3\% & 83.6\% & 89.5\% & 96.7\%     & 55.9\% \\
		50                   & 78.5\% & 86.9\% & 92.8\%     & 32.6\% & 81.5\% & 88.2\%   & 91.7\%     & 46.2\% & 83.5\% & 89.6\% & 93.6\%     & 36.1\% \\
		\hline
	\end{tabular}
	\label{tab:experiment_vertical}
\end{table*}

\subsection{Preserving Privacy of Attributes}
\label{sect:privacy_protection2}

In this section, we evaluate ARS on a stronger assumption that users have some private attributes to protect. We assess the effectiveness of defense against attribute extraction attacks by how close the extracted features are to a fixed vector given by users. The experiments are conducted on CelebA. For all participants, we set predicting attribute "High Cheekbones" as the collaborative learning task, while selecting "Male" and "Smiling" as private attributes. Then we let each user train attribute extraction network $F_k$ corresponding to the $k$-th private attribute, which is similar to the adv-training decoder attack. We choose some typical values of $\lambda$ to generate adversarial noise, and set the fixed vector $r=(1,0)$ since the outputs of classifiers are two-dimensional vectors. Then we record the proportion that the predicted private attribute $F_k(\hat{z})$ is equal to $r$ to estimate the ability of ARS to mislead attribute extraction models.

\begin{figure}
	\centering
	\subfloat[$\lambda_D=1$, $\lambda_1=\lambda_2=0$]{
		\includegraphics[width=0.45\linewidth]{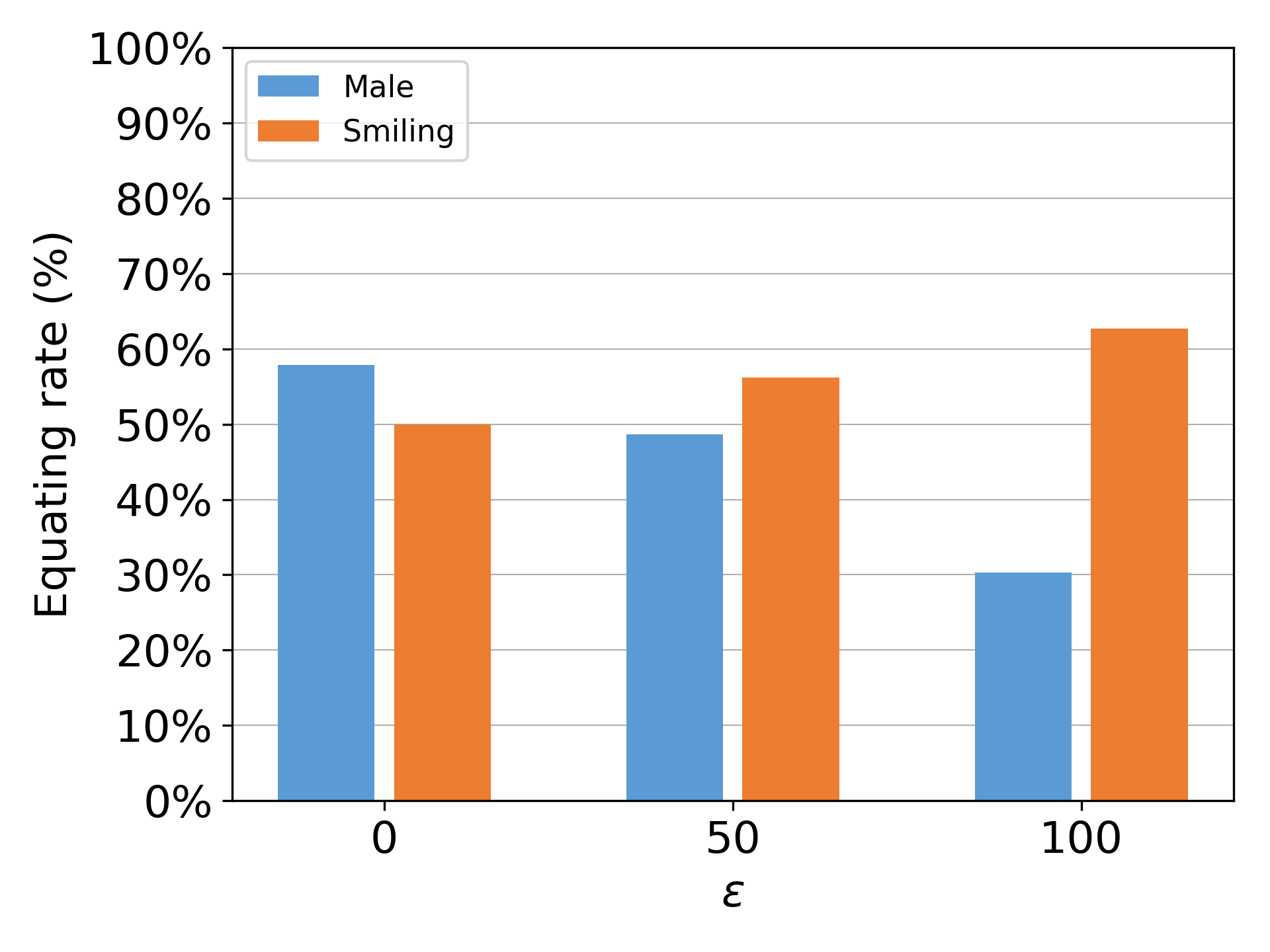}
	}
	\subfloat[$\lambda_D=0.5$, $\lambda_1=\lambda_2=0.25$]{
		\includegraphics[width=0.45\linewidth]{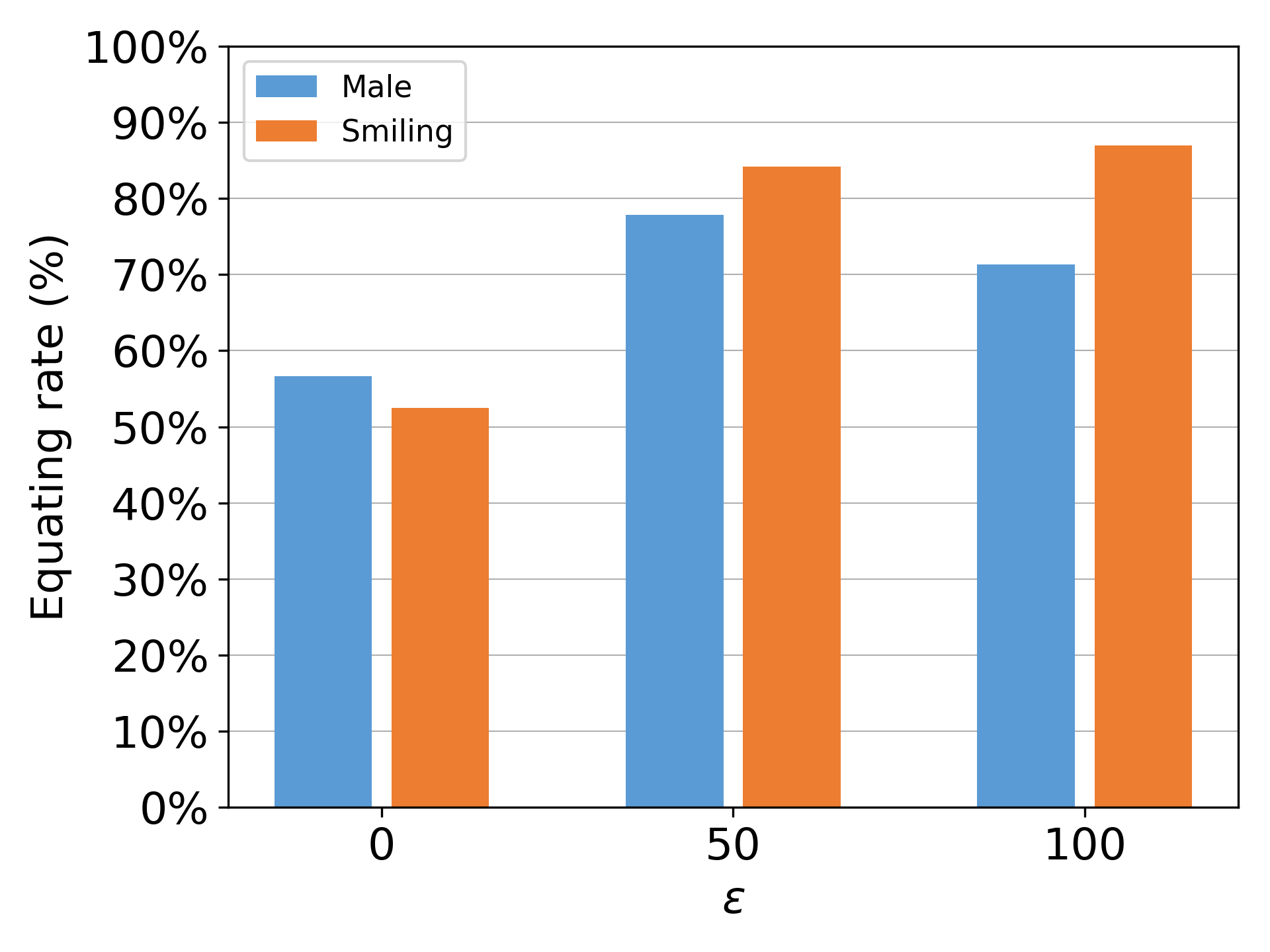}
	}
	
	\subfloat[$\lambda_D=\lambda_1=\lambda_2=0.33$]{
		\includegraphics[width=0.45\linewidth]{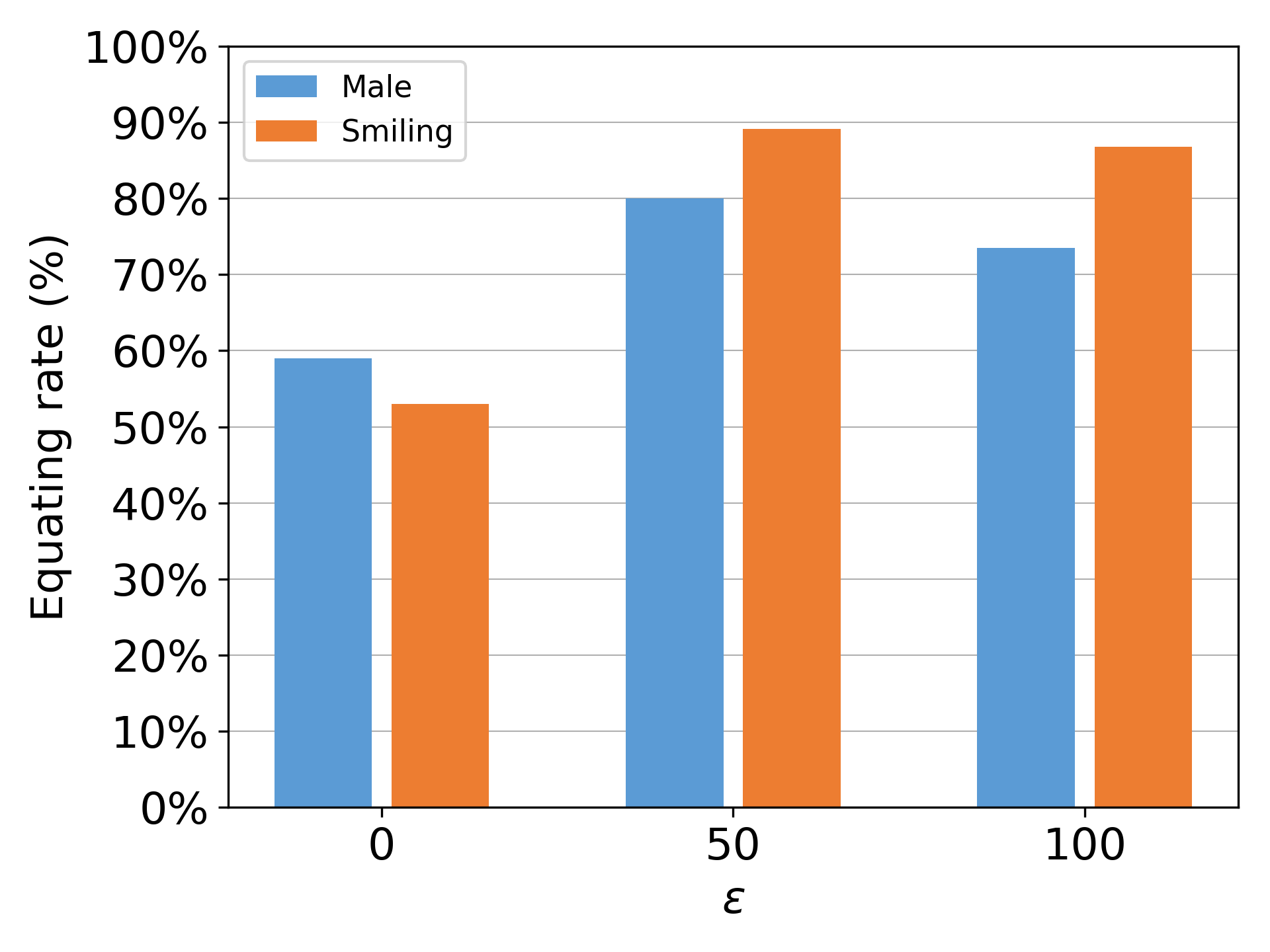}
	}
	\subfloat[$\lambda_D=0$, $\lambda_1=\lambda_2=0.5$]{
		\includegraphics[width=0.45\linewidth]{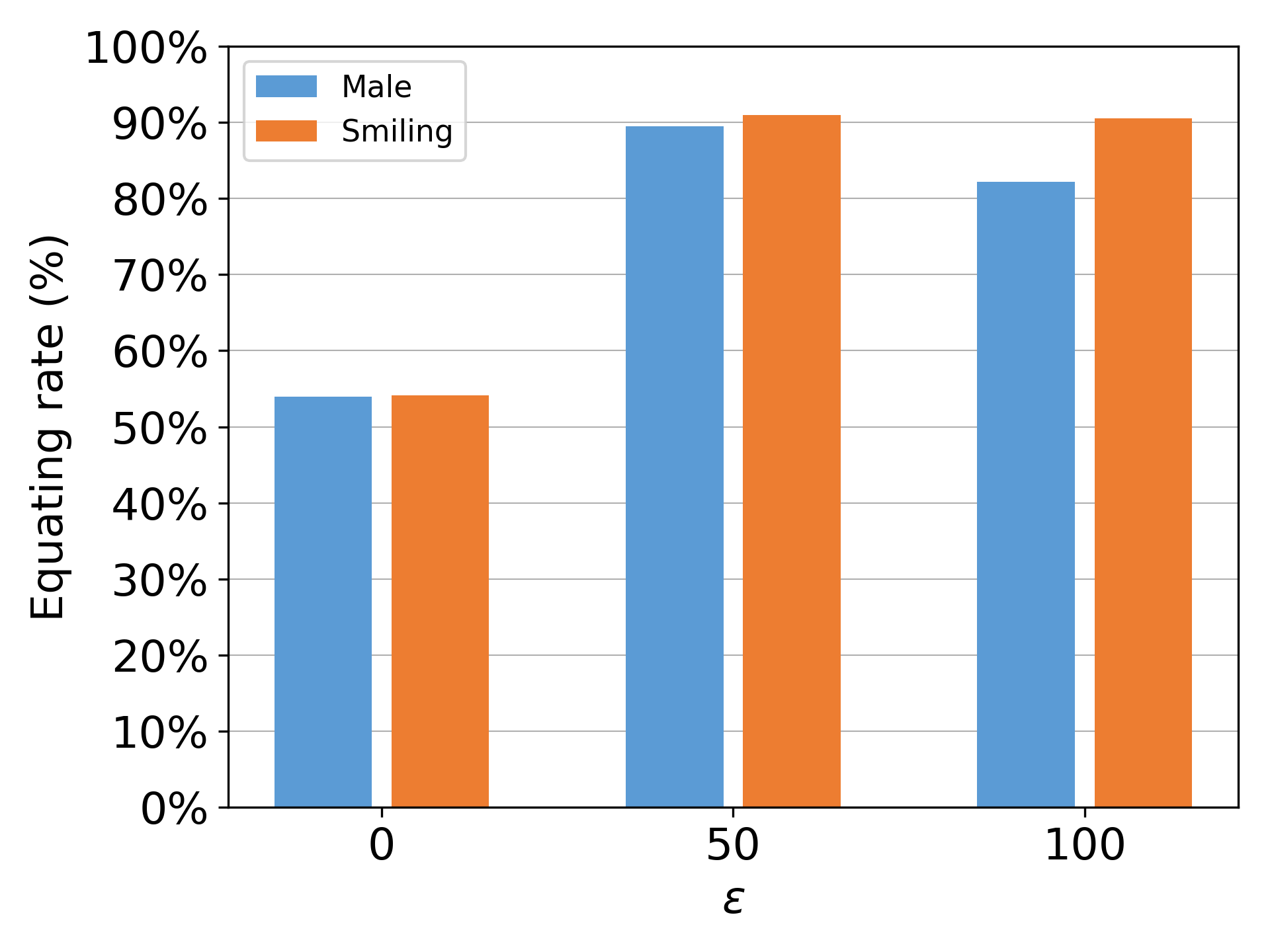}
	}
	
	\subfloat[$\lambda_D=0$, $\lambda_1=1$, $\lambda_2=0$]{
		\includegraphics[width=0.45\linewidth]{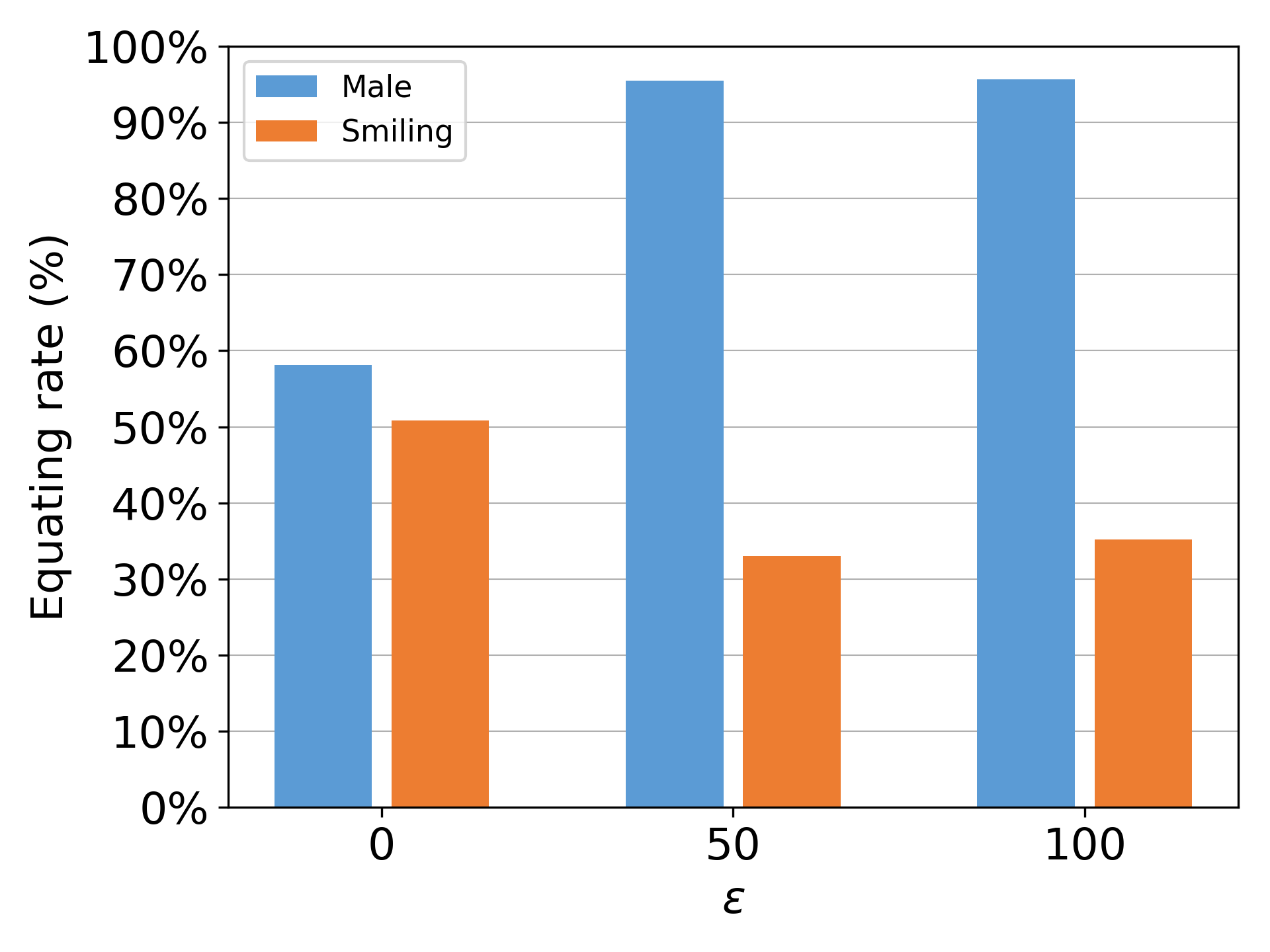}
	}
	\subfloat[$\lambda_D=\lambda_1=0.5$, $\lambda_2=0$]{
		\includegraphics[width=0.45\linewidth]{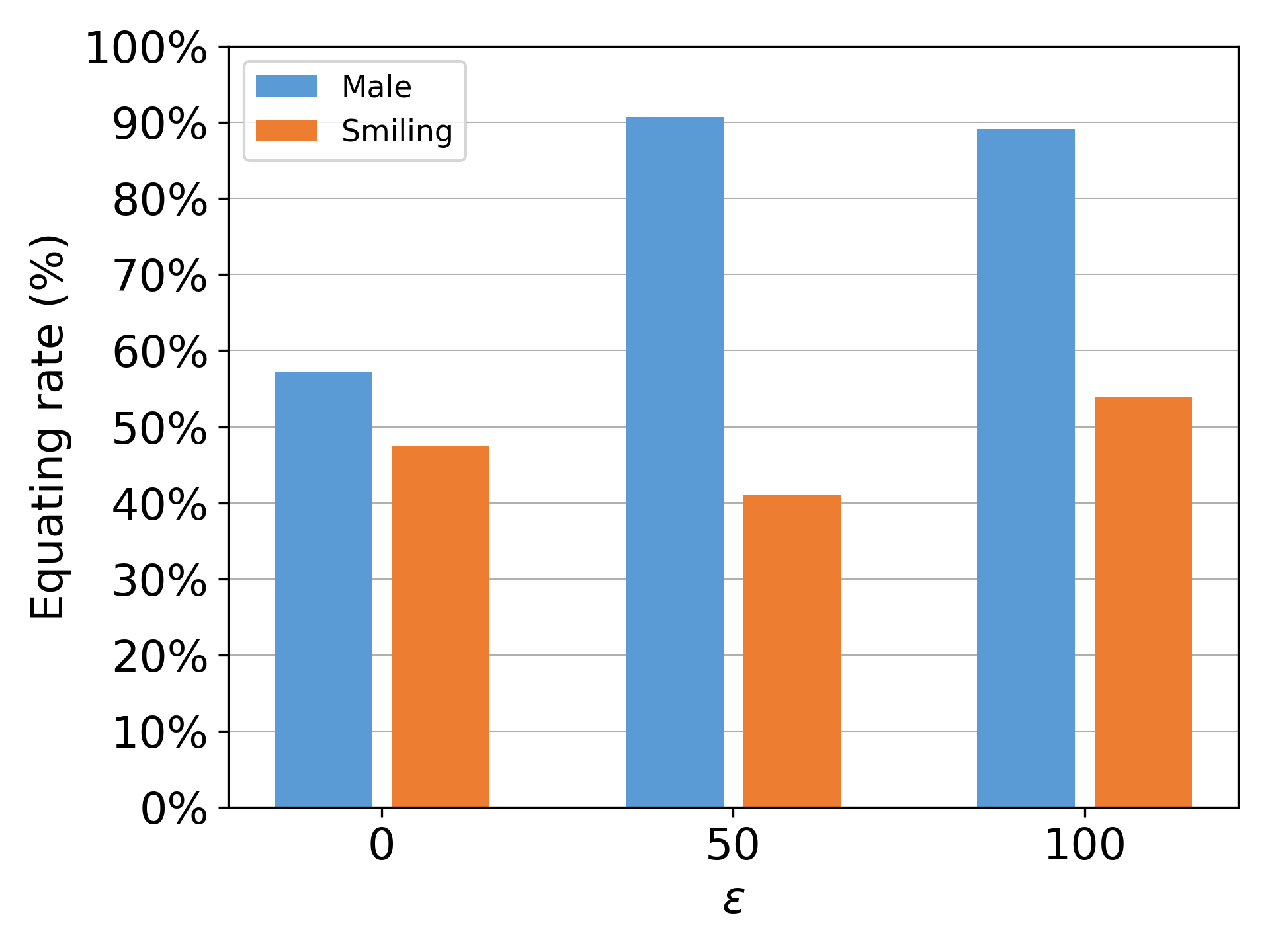}
	}	
	\caption{Proportion that the predictions of private attributes are equal to a given fixed vector.}
	\label{figure:exp_attributes_1}
\end{figure}

We analyze the effect of different compositions of adversarial noise by changing $\lambda$ as illustrated in Figure \ref{figure:exp_attributes_1}. (a-d) show that with the increase of $\lambda_k$, the probability that predictions of the $k$-th attribute are equal to $r$ becomes higher. Note that sometimes the equating rate gets lower when $\epsilon$ increases to $100$, this may be caused by influence of the other components of adversarial noise. (e-f) demonstrate that the weight of a noise $\delta_k$ has a great effect on the privacy of the $k$-th attribute. Our further experiment shows that the accuracy of attack classifiers can be close to $50\%$ when $\epsilon=50$.

We further show that the accuracy of classifiers for preserved attributes is close to $50\%$. Figure \ref{figure:exp_attributes_2} presents the accuracy of classifiers for three attributes, two of which are private. With the increase of $\lambda_k$, the accuracy corresponding to the $k$-th attribute approaches $50\%$ if $\lambda_k>0$. In addition, when $\epsilon$ is set to $100$, the classification accuracy of the second private attribute "Smiling" has a significant drop even if $\lambda_2=0$, this again proves the trade-off between utility and privacy of shared data.In summary, ARS is shown to be effective against attribute extraction attacks with acceptable privacy budget $\epsilon$.

\begin{figure}[htbp]
	\centering
	\subfloat[$\lambda_D=1$, $\lambda_1=\lambda_2=0$]{
		\includegraphics[width=0.45\linewidth]{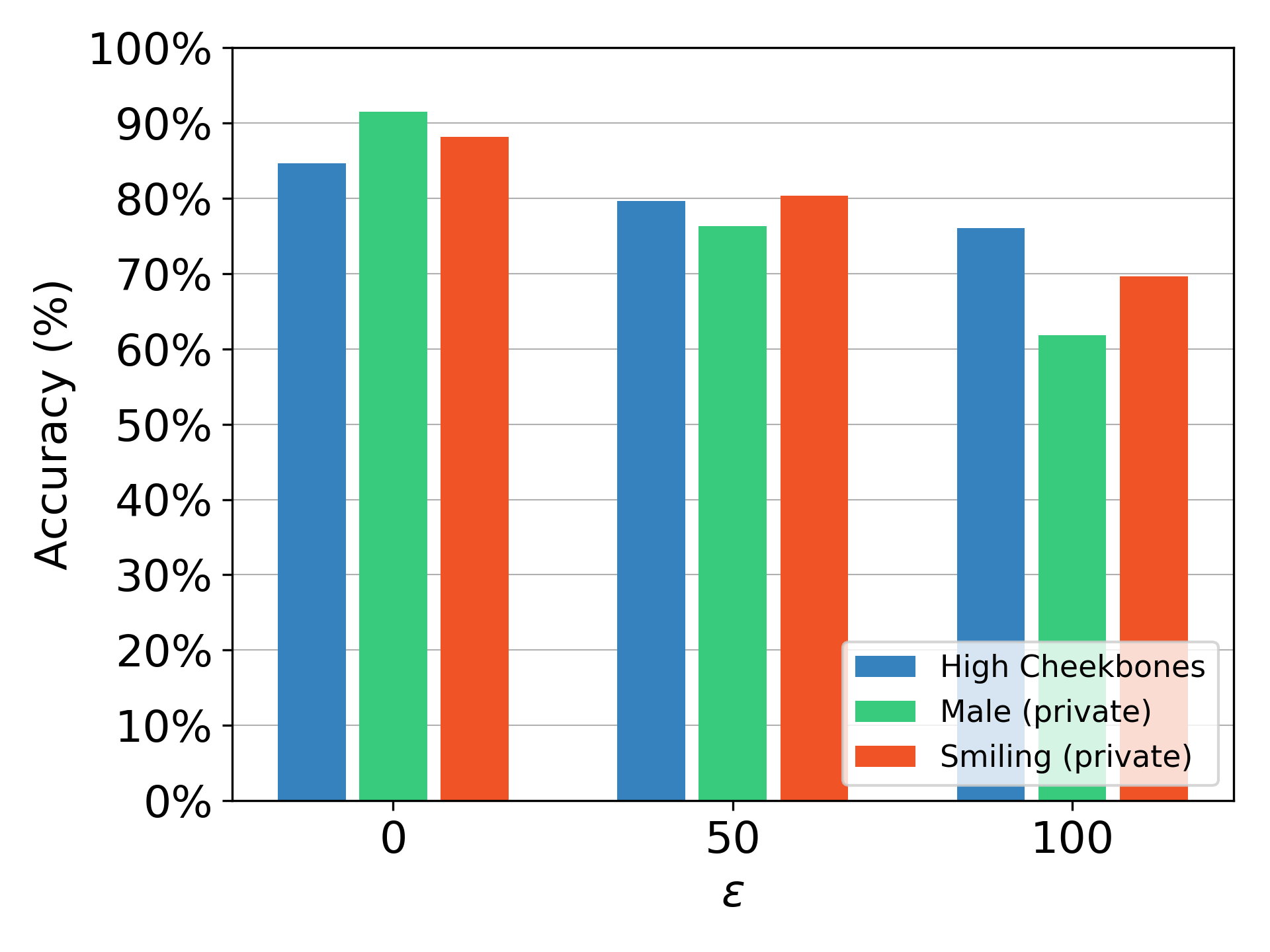}
	}%
	\subfloat[$\lambda_D=0.5$, $\lambda_1=\lambda_2=0.25$]{
		\includegraphics[width=0.45\linewidth]{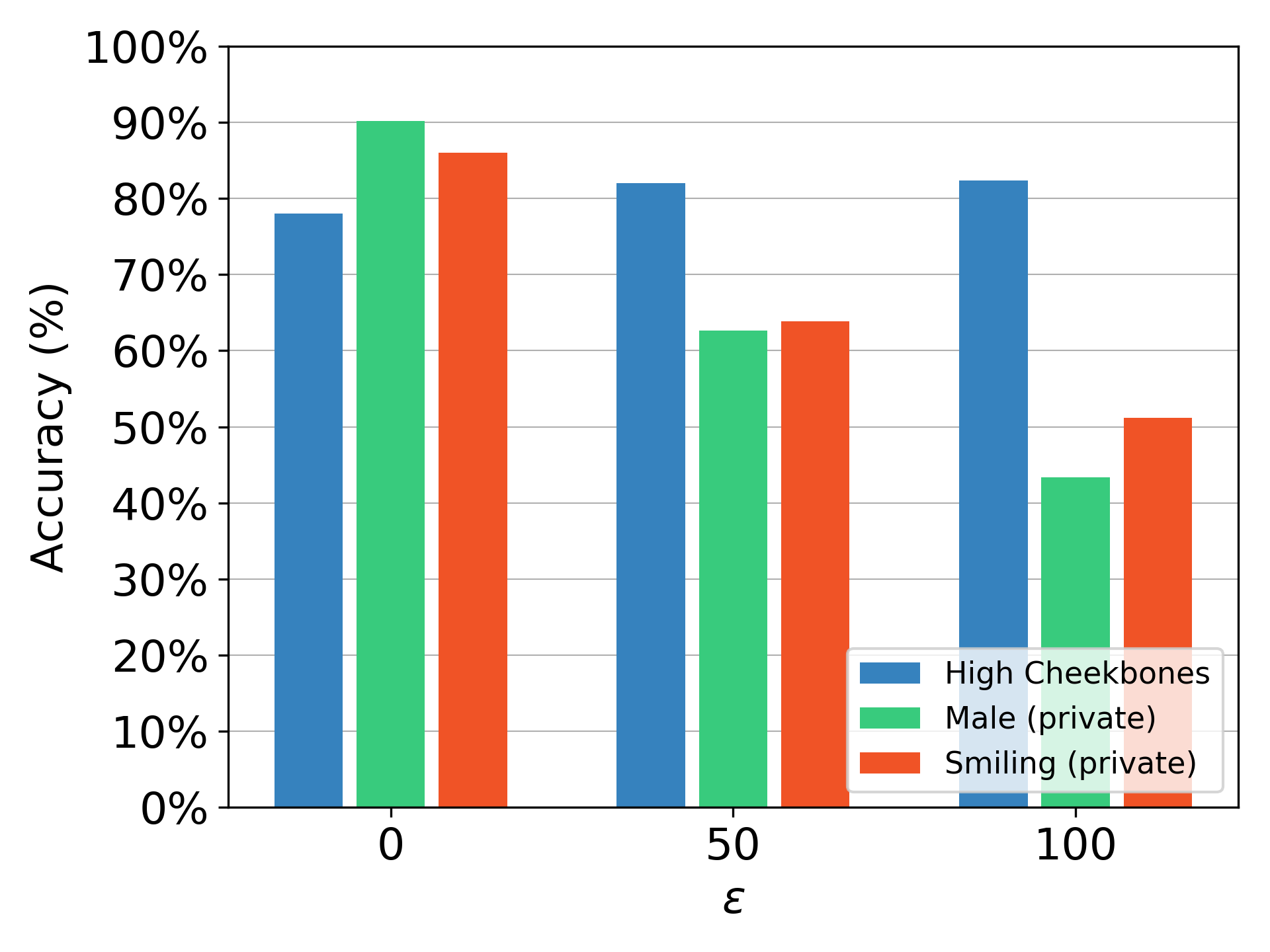}
	}%
	
	\subfloat[$\lambda_D=\lambda_1=\lambda_2=0.33$]{
		\includegraphics[width=0.45\linewidth]{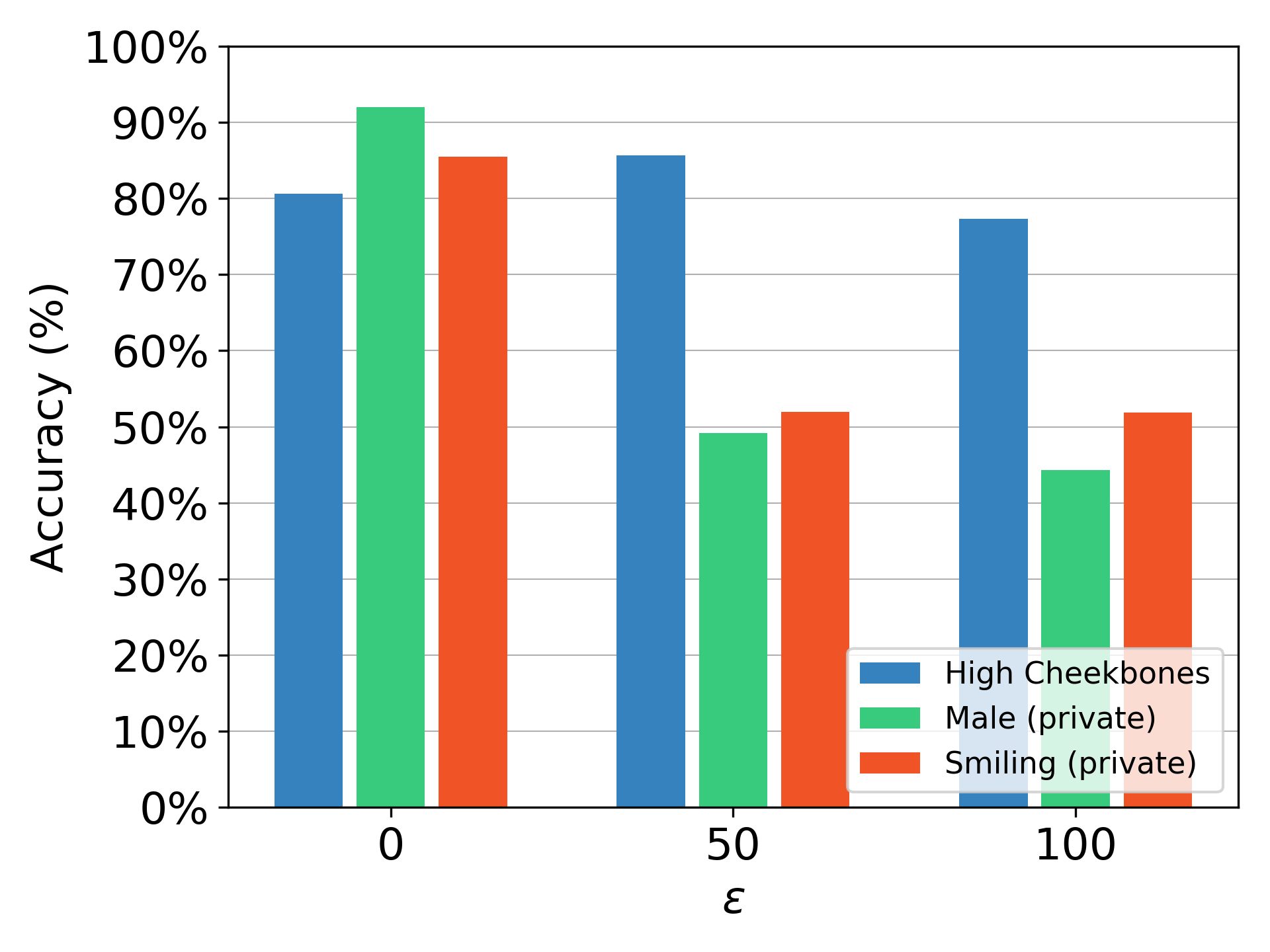}
	}%
	\subfloat[$\lambda_D=0$, $\lambda_1=\lambda_2=0.5$]{
		\includegraphics[width=0.45\linewidth]{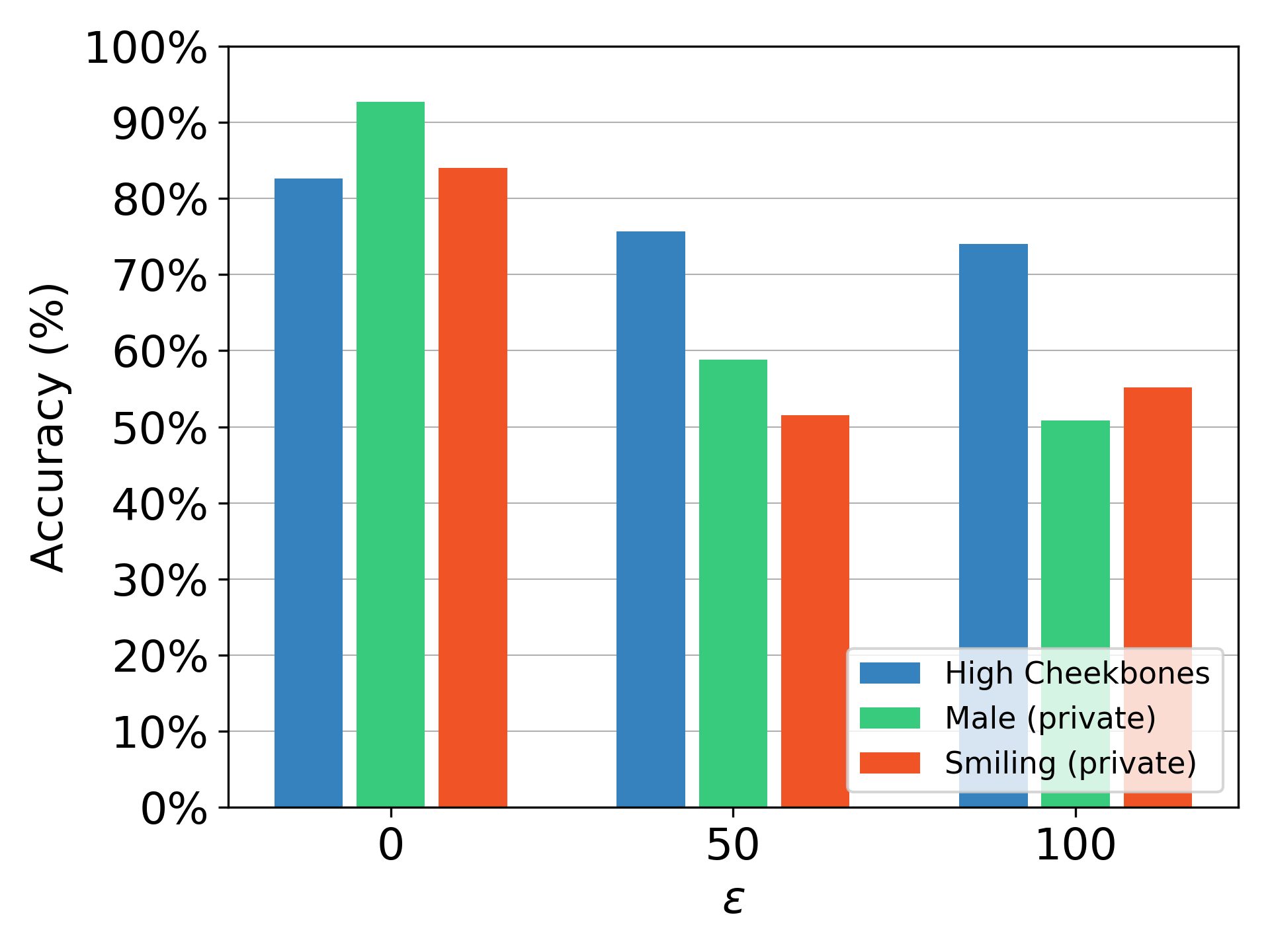}
	}%
	
	\subfloat[$\lambda_D=0$, $\lambda_1=1$, $\lambda_2=0$]{
		\includegraphics[width=0.45\linewidth]{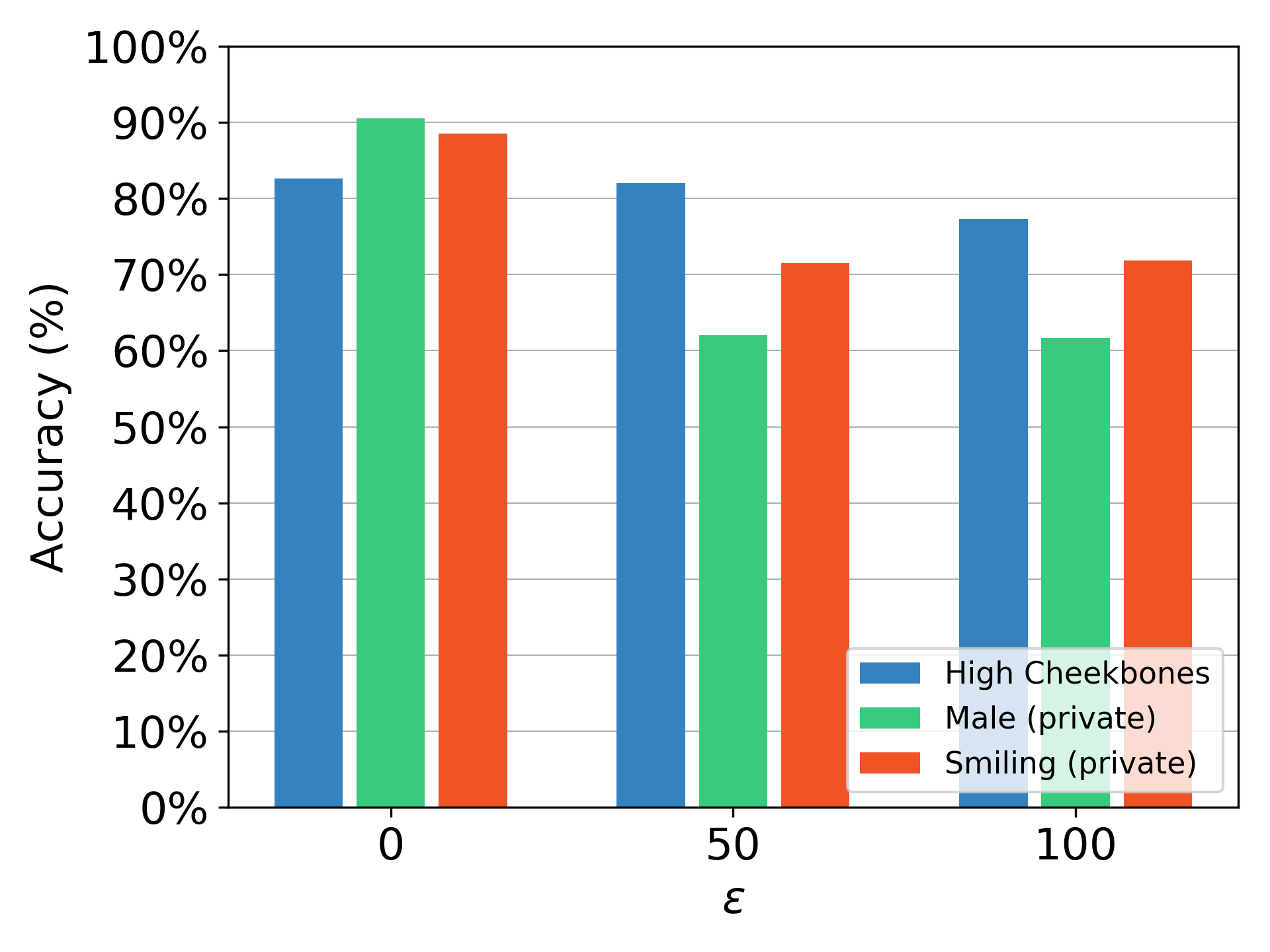}
	}%
	\subfloat[$\lambda_D=0.5$, $\lambda_1=0.5$, $\lambda_2=0$]{
		\includegraphics[width=0.45\linewidth]{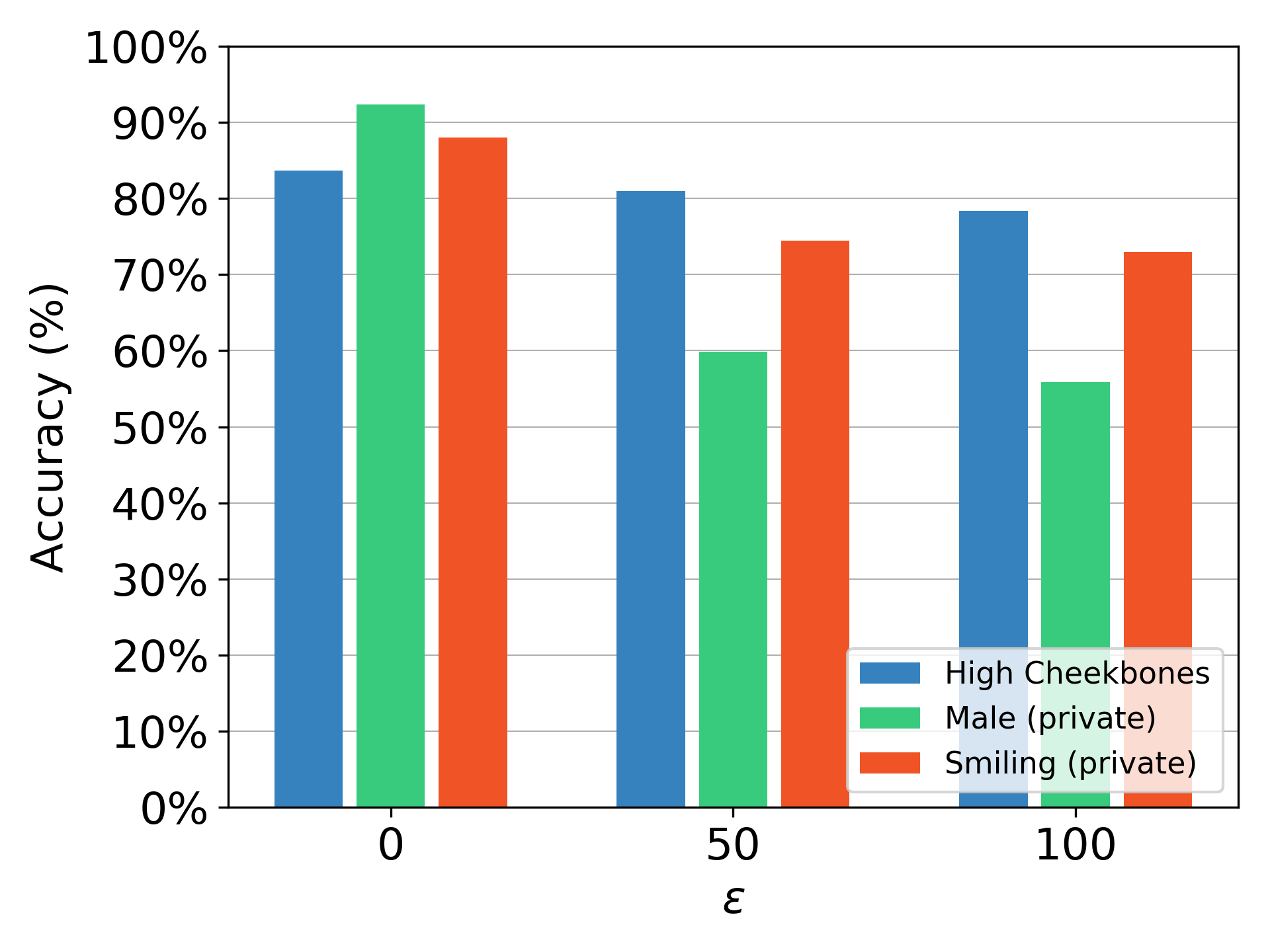}
	}%
	\centering
	\caption{Classification accuracy on three attributes, with variable value of $\epsilon$ and $\lambda$. "Male" and "Smiling" are private attributes.}
	\label{figure:exp_attributes_2}
\end{figure}

\begin{table}[htbp]
	\centering
	\caption{PSNR of different composition of noise and $\epsilon$.}
	\label{tab:experiment4_2}
	\begin{tabular}{cccccc}
		\toprule
		$(\lambda_D, \lambda_1, \lambda_2)$&
		$\epsilon=0$ & $\epsilon=25$ & $\epsilon=50$ & $\epsilon=75$ & $\epsilon=100$ \\ \midrule
		$(1.0, 0.0, 0.0)$ & 22.617 & 15.274 & \textbf{8.812} & 7.456 & 6.720 \\
		$(0.5, 0.25, 0.25)$ & 22.492 & 18.363 & \textbf{12.033} & 10.261 & 9.522 \\
		$(0.33, 0.33, 0.33)$ & 22.497 & 19.408 & 14.473 & 11.976 & 10.671 \\
		$(0.0, 0.5, 0.5)$ & 22.482 & 20.827 & 19.109 & 15.004 & 13.179 \\
		\bottomrule
	\end{tabular}
\end{table}

We next evaluate the reconstruction error under the same scenario. As we can see in Table \ref{tab:experiment4_2}, larger $\epsilon$ and $\lambda_D$ lead to greater defense against reconstruction attacks. If we consider $PSNR=12.033$ an acceptable privacy leakage since it is smaller than the results of similar representation sharing works \cite{xiao2019adversarial} and \cite{Ferdowsi2020PrivacyPreservingIS} we compared in Section \ref{sect:comparison}, then $\lambda_D=\frac{1}{2}, \lambda_k=\frac{1}{2M}$ is a good choice to defend against reconstruction and attribute extraction attacks at the same time.

\subsection{Task-Independence Study}
\label{sect:task_independence_study}

Another advantage of ARS mechanism is that the encoder publishing phase is independent of the collaborative learning phase. Most of the up-to-date joint learning frameworks are task-oriented. For example, in the FederatedAveraging algorithm, a server builds a global model according to a specific task, and communicate the parameters with clients. For prior data sharing frameworks, the networks to extract latent representations of data are usually trained with the task-oriented models (such as classifiers). The accuracy of models trained from these representations is regarded as a part of optimization objective functions of feature extracting networks. Task-dependence causes low data utilization. If parties in an existing collaborative learning framework have a new deep learning task, they have to build another framework and train feature extracting networks once again, which results in a heavy cost to transform data and train deep learning models.

In the ARS mechanism, apart from private attributes that are protected, the representation generating process is independent of deep learning tasks. We set up a series of experiments on CelebA to demonstrate this property. We select five of the forty attributes in CelebA and set a binary classification for each attribute, none of them are private to users. In the collaborative learning phase, participants train five classification networks on the latent representations. Each classifier corresponding to an attribute. Since the MSE loss is only related to the encoder, but not to deep learning tasks, we only focus on classification accuracy. Table \ref{tab:experiment5} shows the results with various scale of noise.

\begin{table}[htbp]
	\centering
	\caption{Classification accuracy in different tasks. For each of the five attributes from CelebA dataset, we train classifiers to predict whether the label is positive or negative from $\hat{z}$ generated by the same public encoder.}
	\label{tab:experiment5}
	\begin{tabular}{cccccc} \toprule
		Attribute
		& $\epsilon=0$ & $\epsilon=25$ & $\epsilon=50$ & $\epsilon=75$ & $\epsilon=100$ \\ \midrule
		Heavy Makeup              & 86.4\% & 85.7\% & 85.0\% & 87.5\% & 85.9\% \\
		High Cheekbones           & 80.6\% & 84.2\% & 80.1\% & 82.8\% & 82.1\% \\
		Male                       & 91.7\% & 91.7\% & 90.0\% & 89.3\% & 88.3\% \\
		Smiling                    & 88.0\% & 84.7\% & 87.9\% & 87.5\% & 84.3\% \\
		Wearing Lipstick          & 90.2\% & 80.6\% & 86.7\% & 85.0\% & 82.8\% \\ \bottomrule
	\end{tabular}
\end{table}

As shown in the results, classification accuracy is higher than 73\%, which indicates acceptable correctness. Only the accuracy of predicting the attribute "Mouth Slightly Open" is lower than 80\% with the increase of $\epsilon$. The results demonstrate that ARS is task-independent. If there are demands for data sharing, the initiator can just train the public autoencoder, without considering how data representations would be used. For participants, the only remarkable thing is to determine their private attributes, and generate adversarial noise to preserve these attributes. This property prevents the utility of latent representations from being limited to specific tasks, and ensures the robustness of shared data. Representations generated by users independently can achieve good performance in various tasks if the value of hyper-parameter $\epsilon$ is chosen well. Consequently, ARS can preserve privacy during the data sharing process, while maintaining the utility of data in collaborative learning.

\section{Discussion}
\label{sect:discussion}
\subsection{Discussions on Noise Masking}
\label{sect:noise_masking}

We focus on the security of noise masking mechanism by studying whether it can defend against brute-force searching attacks. An attacker can randomly enumerate several mask vectors, train inverse models on representations with these mask vectors respectively and take the vector that performs best in reconstructing others' data as a good approximation of the victim's mask. We explore experimentally the relationship between the reconstruction loss $\mathcal{L}_R$ and the overlapping rate of masks held by attackers and defenders, which equals to the Hamming distance of the mask vectors divided by their dimension. The experiment is conducted on MNIST, with settings stated in Section \ref{sect:experiments}. As Table \ref{tab:mask} illustrates, a higher overlapping rate leads to a higher risk of privacy leakage. So we'll next study the overlapping of masks.
\begin{table}[htbp]
	\centering
	\caption{Reconstruction loss with various overlapping rate of masks held by attackers and defenders.}
	\begin{tabular}{c|ccccc}
		\hline
		Overlapping Rate & 0\%	& 25\%	& 50\%	& 75\%	& 100\% \\ \hline
		$\epsilon=50$	& 0.119	& 0.101	& 0.082	& 0.048	& 0.021 \\ \hline
		$\epsilon=100$	& 0.179	& 0.156	& 0.128	& 0.09	& 0.025 \\ \hline
	\end{tabular}
	\label{tab:mask}
\end{table}

For any $n$-dimensional mask vectors $\mathbf{m}_1$ and $\mathbf{m}_2$, we denote the Hamming distance between them as $H(\mathbf{m}_1,\mathbf{m}_2)$, and define the overlapping rate between $\mathbf{m}_1$ and $\mathbf{m}_2$ as $o(\mathbf{m}_1,\mathbf{m}_2)=\frac{n-H(\mathbf{m}_1,\mathbf{m}_2)}{n}$. Then we have 
\begin{equation}
	\mathbb{P}\left [n-H(\mathbf{m}_1,\mathbf{m}_2) = i\right ] = \frac{1}{2^n}\binom{n}{i},
\end{equation}
which means that $X = n-H(\mathbf{m}_1,\mathbf{m}_2) \sim B(n,0.5)$.

Suppose $t$ is a real number such that $\frac{1}{2} <t \leq 1$, then from the De Moivre-Laplace theorem, the probability that $\mathbf{m}_1$ and $\mathbf{m}_2$ have $t \cdot n$ bits different is:
\begin{equation}
	\begin{aligned}
		& \quad \lim_{n \to \infty} \mathbb{P} \left [o(\mathbf{m}_1,\mathbf{m}_2) \geq t \right ] \\
		& = \lim_{n \to \infty} \mathbb{P} \left [tn \leq n-H(\mathbf{m}_1,\mathbf{m}_2) \leq n \right ] \\
		& = \lim_{n \to \infty} \mathbb{P} \left [(2t-1)\sqrt{n} \leq \frac{X-\frac{1}{2}n}{\frac{1}{2}\sqrt{n}} \leq \sqrt{n} \right ] \\
		& = \lim_{n \to \infty} \frac{1}{2\pi} \int_{(2t-1)\sqrt{n}}^{\sqrt{n}} \mathrm{e}^{-\frac{x^2}{2}} \mathrm{d}x \\
		& = \lim_{n \to \infty} \Phi(\sqrt{n}) - \Phi((2t-1)\sqrt{n}) \\
		& = 0.
	\end{aligned}
\end{equation}

Therefore, if the dimension $n$ is large enough, the probability that the overlapping rate of two random $n$-dimensional vectors is larger than $t$ approaches to $0$ for $\forall \frac{1}{2} < t \leq 1$. Moreover, we consider $\epsilon=50$ as an acceptable privacy budget for preserving information of data. That is to say, an attack is considered successful if the overlapping rate of masks held by the attacker and user should be greater than a real number $t$, where $\frac{1}{2} <t \leq 1$. When the dimension of latent representations is large enough, the privacy of users' data can be guaranteed. For example, the dimension of latent representations is $256$. If we accept 75\% as overlapping rate, then we have $ \mathbb{P} \left [o(\mathbf{m}_1,\mathbf{m}_2) \geq 0.75 \right ] \leq 2.449 \times 10^{-16}$, which means that the privacy of data can be considered well preserved by mask mechanism.

\subsection{Future Work}
\label{sect:limitations}

Although ARS shows good performance in our given scenarios, this mechanism still has some limitations. In the horizontal data partitioning scenario, all users generate data representations through the same feature extraction network, which is called the common encoder. This requires the selected initiator to have a sufficient amount of representative data. In practice, however, participants in joint learning may lack enough training samples, or the data of each user may not be independent and identically distributed (non-IID) \cite{kairouz2021advances}. This leads the common encoder to be overfitted, so that it will no longer be valid for all users. To cope with this problem, we tried to let the parties train their own encoders on the local datasets, while ensuring that the latent representations have the same distribution. For example, each user applies variational autoencoder (VAE) \cite{kingma2013auto} to constrain data representations to the standard normal distribution. Nevertheless, it is difficult to guarantee that the same dimension of representations generated by different encoders expresses the same semantic. Therefore, aggregation of the shared data representations will no longer make sense.

In this study, we relax the hypothesis of the data owners, requiring at least one party has sufficient training samples, and the samples of each participant have identical distribution. This assumption in accordance with the realistic B2C (business to customer) settings, where the initiator can be an enterprise with a certain accumulation of data. It can initiate data sharing with individual users and provide pre-trained feature extraction models to them. Future work will be dedicated to collaborative learning on non-IID data, and we believe domain adaptation of data representations is a viable solution to this problem.

We also introduce task-independence of the shared data, and how this property can help to reduce communication cost. In this work, data representations are extracted by unsupervised autoencoders to avoid task-orientation, yet some prior knowledge such as the available labels of data is underutilized. A possible area of future work is multi-task learning \cite{zhang2021survey}, where the given tasks or training labels can be made full use of and contribute to each client’s different local problems. Related studies may bring higher utility of data to task-independent collaborative learning.

\section{Conclusion}
\label{sect:conclusion}

In this work, we propose ARS, a privacy-preserving collaborative learning framework. Users share representations of data to train downstream models. Adversarial noise is used to protect shared data from model inversion attacks. We evaluate our mechanism and demonstrate that adding masked adversarial noise on latent representations has a great effect in defending against reconstruction and attribute extraction attacks, while maintaining almost the same utility as MPC and FL based training. Compared with some prior data sharing mechanisms, ARS outperforms them in privacy preservation. Besides, ARS is task-independent, and requires no centralized control. Our work can be applied to collaborative learning scenarios, and provides a new idea on the research of data sharing and joint learning frameworks.
\bibliographystyle{ACM-Reference-Format}
\bibliography{sample-base}


\begin{thebibliography}{33}


\ifx \showCODEN    \undefined \def \showCODEN     #1{\unskip}     \fi
\ifx \showDOI      \undefined \def \showDOI       #1{#1}\fi
\ifx \showISBNx    \undefined \def \showISBNx     #1{\unskip}     \fi
\ifx \showISBNxiii \undefined \def \showISBNxiii  #1{\unskip}     \fi
\ifx \showISSN     \undefined \def \showISSN      #1{\unskip}     \fi
\ifx \showLCCN     \undefined \def \showLCCN      #1{\unskip}     \fi
\ifx \shownote     \undefined \def \shownote      #1{#1}          \fi
\ifx \showarticletitle \undefined \def \showarticletitle #1{#1}   \fi
\ifx \showURL      \undefined \def \showURL       {\relax}        \fi
\providecommand\bibfield[2]{#2}
\providecommand\bibinfo[2]{#2}
\providecommand\natexlab[1]{#1}
\providecommand\showeprint[2][]{arXiv:#2}

\bibitem[\protect\citeauthoryear{Agrawal, Shahin~Shamsabadi, Kusner, and
  Gasc{\'o}n}{Agrawal et~al\mbox{.}}{2019}]%
        {agrawal2019quotient}
\bibfield{author}{\bibinfo{person}{Nitin Agrawal}, \bibinfo{person}{Ali
  Shahin~Shamsabadi}, \bibinfo{person}{Matt~J Kusner}, {and}
  \bibinfo{person}{Adri{\`a} Gasc{\'o}n}.} \bibinfo{year}{2019}\natexlab{}.
\newblock \showarticletitle{QUOTIENT: two-party secure neural network training
  and prediction}. In \bibinfo{booktitle}{\emph{Proceedings of the 2019 ACM
  SIGSAC Conference on Computer and Communications Security}}.
  \bibinfo{pages}{1231--1247}.
\newblock


\bibitem[\protect\citeauthoryear{Bengio, Courville, and Vincent}{Bengio
  et~al\mbox{.}}{2013}]%
        {bengio2013representation}
\bibfield{author}{\bibinfo{person}{Yoshua Bengio}, \bibinfo{person}{Aaron
  Courville}, {and} \bibinfo{person}{Pascal Vincent}.}
  \bibinfo{year}{2013}\natexlab{}.
\newblock \showarticletitle{Representation learning: A review and new
  perspectives}.
\newblock \bibinfo{journal}{\emph{IEEE transactions on pattern analysis and
  machine intelligence}} \bibinfo{volume}{35}, \bibinfo{number}{8}
  (\bibinfo{year}{2013}), \bibinfo{pages}{1798--1828}.
\newblock


\bibitem[\protect\citeauthoryear{Church}{Church}{2017}]%
        {church2017word2vec}
\bibfield{author}{\bibinfo{person}{Kenneth~Ward Church}.}
  \bibinfo{year}{2017}\natexlab{}.
\newblock \showarticletitle{Word2Vec}.
\newblock \bibinfo{journal}{\emph{Natural Language Engineering}}
  \bibinfo{volume}{23}, \bibinfo{number}{1} (\bibinfo{year}{2017}),
  \bibinfo{pages}{155--162}.
\newblock


\bibitem[\protect\citeauthoryear{Dua and Graff}{Dua and Graff}{2017}]%
        {Dua2019}
\bibfield{author}{\bibinfo{person}{Dheeru Dua} {and} \bibinfo{person}{Casey
  Graff}.} \bibinfo{year}{2017}\natexlab{}.
\newblock \bibinfo{title}{{UCI} Machine Learning Repository}.
\newblock
\newblock
\urldef\tempurl%
\url{http://archive.ics.uci.edu/ml}
\showURL{%
\tempurl}


\bibitem[\protect\citeauthoryear{Erhan, Courville, Bengio, and Vincent}{Erhan
  et~al\mbox{.}}{2010}]%
        {erhan2010does}
\bibfield{author}{\bibinfo{person}{Dumitru Erhan}, \bibinfo{person}{Aaron
  Courville}, \bibinfo{person}{Yoshua Bengio}, {and} \bibinfo{person}{Pascal
  Vincent}.} \bibinfo{year}{2010}\natexlab{}.
\newblock \showarticletitle{Why does unsupervised pre-training help deep
  learning?}. In \bibinfo{booktitle}{\emph{Proceedings of the thirteenth
  international conference on artificial intelligence and statistics}}. JMLR
  Workshop and Conference Proceedings, \bibinfo{pages}{201--208}.
\newblock


\bibitem[\protect\citeauthoryear{Ferdowsi, Razeghi, Holotyak, Calmon, and
  Voloshynovskiy}{Ferdowsi et~al\mbox{.}}{2020}]%
        {Ferdowsi2020PrivacyPreservingIS}
\bibfield{author}{\bibinfo{person}{Sohrab Ferdowsi}, \bibinfo{person}{Behrooz
  Razeghi}, \bibinfo{person}{Taras Holotyak}, \bibinfo{person}{Flavio~P.
  Calmon}, {and} \bibinfo{person}{Slava Voloshynovskiy}.}
  \bibinfo{year}{2020}\natexlab{}.
\newblock \showarticletitle{Privacy-Preserving Image Sharing via Sparsifying
  Layers on Convolutional Groups}.
\newblock \bibinfo{journal}{\emph{ICASSP}} (\bibinfo{year}{2020}).
\newblock


\bibitem[\protect\citeauthoryear{Gentry}{Gentry}{2009}]%
        {gentry2009fully}
\bibfield{author}{\bibinfo{person}{Craig Gentry}.}
  \bibinfo{year}{2009}\natexlab{}.
\newblock \showarticletitle{Fully homomorphic encryption using ideal lattices}.
  In \bibinfo{booktitle}{\emph{Proceedings of the forty-first annual ACM
  symposium on Theory of computing}}. \bibinfo{pages}{169--178}.
\newblock


\bibitem[\protect\citeauthoryear{{Geyer}, {Klein}, and {Nabi}}{{Geyer}
  et~al\mbox{.}}{2017}]%
        {geyer2017differentially}
\bibfield{author}{\bibinfo{person}{Robin~C. {Geyer}},
  \bibinfo{person}{Tassilo~J. {Klein}}, {and} \bibinfo{person}{Moin {Nabi}}.}
  \bibinfo{year}{2017}\natexlab{}.
\newblock \showarticletitle{Differentially Private Federated Learning: A Client
  Level Perspective}.
\newblock \bibinfo{journal}{\emph{arXiv preprint arXiv:1712.07557}}
  (\bibinfo{year}{2017}).
\newblock


\bibitem[\protect\citeauthoryear{Goodfellow, Bengio, and Courville}{Goodfellow
  et~al\mbox{.}}{2016}]%
        {goodfellow2016deep}
\bibfield{author}{\bibinfo{person}{Ian Goodfellow}, \bibinfo{person}{Yoshua
  Bengio}, {and} \bibinfo{person}{Aaron Courville}.}
  \bibinfo{year}{2016}\natexlab{}.
\newblock \bibinfo{booktitle}{\emph{Deep learning}}.
\newblock \bibinfo{publisher}{MIT press}.
\newblock


\bibitem[\protect\citeauthoryear{Goodfellow, Shlens, and Szegedy}{Goodfellow
  et~al\mbox{.}}{2015}]%
        {goodfellow2014explaining}
\bibfield{author}{\bibinfo{person}{Ian Goodfellow}, \bibinfo{person}{Jonathon
  Shlens}, {and} \bibinfo{person}{Christian Szegedy}.}
  \bibinfo{year}{2015}\natexlab{}.
\newblock \showarticletitle{Explaining and harnessing adversarial examples}.
\newblock \bibinfo{journal}{\emph{ICLR}} (\bibinfo{year}{2015}).
\newblock


\bibitem[\protect\citeauthoryear{Hard, Rao, Mathews, Ramaswamy, Beaufays,
  Augenstein, Eichner, Kiddon, and Ramage}{Hard et~al\mbox{.}}{2018}]%
        {hard2018federated}
\bibfield{author}{\bibinfo{person}{Andrew Hard}, \bibinfo{person}{Kanishka
  Rao}, \bibinfo{person}{Rajiv Mathews}, \bibinfo{person}{Swaroop Ramaswamy},
  \bibinfo{person}{Fran{\c{c}}oise Beaufays}, \bibinfo{person}{Sean
  Augenstein}, \bibinfo{person}{Hubert Eichner}, \bibinfo{person}{Chlo{\'e}
  Kiddon}, {and} \bibinfo{person}{Daniel Ramage}.}
  \bibinfo{year}{2018}\natexlab{}.
\newblock \showarticletitle{Federated learning for mobile keyboard prediction}.
\newblock \bibinfo{journal}{\emph{arXiv preprint arXiv:1811.03604}}
  (\bibinfo{year}{2018}).
\newblock


\bibitem[\protect\citeauthoryear{He, Zhang, and Lee}{He et~al\mbox{.}}{2019}]%
        {he2019model}
\bibfield{author}{\bibinfo{person}{Zecheng He}, \bibinfo{person}{Tianwei
  Zhang}, {and} \bibinfo{person}{Ruby~B Lee}.} \bibinfo{year}{2019}\natexlab{}.
\newblock \showarticletitle{Model inversion attacks against collaborative
  inference}. In \bibinfo{booktitle}{\emph{Proceedings of the 35th Annual
  Computer Security Applications Conference}}. \bibinfo{pages}{148--162}.
\newblock


\bibitem[\protect\citeauthoryear{{Hitaj}, {Ateniese}, and {Perez-Cruz}}{{Hitaj}
  et~al\mbox{.}}{2017}]%
        {hitaj2017deep}
\bibfield{author}{\bibinfo{person}{Briland {Hitaj}}, \bibinfo{person}{Giuseppe
  {Ateniese}}, {and} \bibinfo{person}{Fernando {Perez-Cruz}}.}
  \bibinfo{year}{2017}\natexlab{}.
\newblock \showarticletitle{Deep Models Under the GAN: Information Leakage from
  Collaborative Deep Learning}. In \bibinfo{booktitle}{\emph{Proceedings of the
  2017 ACM SIGSAC Conference on Computer and Communications Security}}.
  \bibinfo{pages}{603--618}.
\newblock


\bibitem[\protect\citeauthoryear{Hore and Ziou}{Hore and Ziou}{2010}]%
        {hore2010image}
\bibfield{author}{\bibinfo{person}{Alain Hore} {and} \bibinfo{person}{Djemel
  Ziou}.} \bibinfo{year}{2010}\natexlab{}.
\newblock \showarticletitle{Image quality metrics: PSNR vs. SSIM}. In
  \bibinfo{booktitle}{\emph{2010 20th international conference on pattern
  recognition}}. IEEE, \bibinfo{pages}{2366--2369}.
\newblock


\bibitem[\protect\citeauthoryear{Kairouz, McMahan, Avent, Bellet, Bennis,
  Bhagoji, Bonawitz, Charles, Cormode, Cummings, et~al\mbox{.}}{Kairouz
  et~al\mbox{.}}{2021}]%
        {kairouz2021advances}
\bibfield{author}{\bibinfo{person}{Peter Kairouz}, \bibinfo{person}{H~Brendan
  McMahan}, \bibinfo{person}{Brendan Avent}, \bibinfo{person}{Aur{\'e}lien
  Bellet}, \bibinfo{person}{Mehdi Bennis}, \bibinfo{person}{Arjun~Nitin
  Bhagoji}, \bibinfo{person}{Kallista Bonawitz}, \bibinfo{person}{Zachary
  Charles}, \bibinfo{person}{Graham Cormode}, \bibinfo{person}{Rachel
  Cummings}, {et~al\mbox{.}}} \bibinfo{year}{2021}\natexlab{}.
\newblock \showarticletitle{Advances and open problems in federated learning}.
\newblock \bibinfo{journal}{\emph{Foundations and Trends{\textregistered} in
  Machine Learning}} \bibinfo{volume}{14}, \bibinfo{number}{1--2}
  (\bibinfo{year}{2021}), \bibinfo{pages}{1--210}.
\newblock


\bibitem[\protect\citeauthoryear{Kingma and Welling}{Kingma and
  Welling}{2013}]%
        {kingma2013auto}
\bibfield{author}{\bibinfo{person}{Diederik~P Kingma} {and}
  \bibinfo{person}{Max Welling}.} \bibinfo{year}{2013}\natexlab{}.
\newblock \showarticletitle{Auto-encoding variational bayes}.
\newblock \bibinfo{journal}{\emph{arXiv preprint arXiv:1312.6114}}
  (\bibinfo{year}{2013}).
\newblock


\bibitem[\protect\citeauthoryear{LeCun}{LeCun}{1998}]%
        {mnist}
\bibfield{author}{\bibinfo{person}{Yann LeCun}.}
  \bibinfo{year}{1998}\natexlab{}.
\newblock \bibinfo{title}{The mnist database of handwritten digits}.
\newblock \bibinfo{howpublished}{\url{http://yann.lecun.com/exdb/mnist/}}.
\newblock


\bibitem[\protect\citeauthoryear{{Liu}, {Chen}, {Liu}, and {Song}}{{Liu}
  et~al\mbox{.}}{2017}]%
        {liu2017delving}
\bibfield{author}{\bibinfo{person}{Yanpei {Liu}}, \bibinfo{person}{Xinyun
  {Chen}}, \bibinfo{person}{Chang {Liu}}, {and} \bibinfo{person}{Dawn {Song}}.}
  \bibinfo{year}{2017}\natexlab{}.
\newblock \showarticletitle{Delving into Transferable Adversarial Examples and
  Black-box Attacks}. In \bibinfo{booktitle}{\emph{ICLR 2017 : International
  Conference on Learning Representations 2017}}.
\newblock


\bibitem[\protect\citeauthoryear{Liu, Luo, Wang, and Tang}{Liu
  et~al\mbox{.}}{2014}]%
        {Liu2014DeepLF}
\bibfield{author}{\bibinfo{person}{Ziwei Liu}, \bibinfo{person}{Ping Luo},
  \bibinfo{person}{Xiaogang Wang}, {and} \bibinfo{person}{Xiaoou Tang}.}
  \bibinfo{year}{2014}\natexlab{}.
\newblock \showarticletitle{Deep Learning Face Attributes in the Wild}.
\newblock \bibinfo{journal}{\emph{2015 IEEE International Conference on
  Computer Vision (ICCV)}} (\bibinfo{year}{2014}), \bibinfo{pages}{3730--3738}.
\newblock


\bibitem[\protect\citeauthoryear{Mahendran and Vedaldi}{Mahendran and
  Vedaldi}{2015}]%
        {mahendran2015understanding}
\bibfield{author}{\bibinfo{person}{Aravindh Mahendran} {and}
  \bibinfo{person}{Andrea Vedaldi}.} \bibinfo{year}{2015}\natexlab{}.
\newblock \showarticletitle{Understanding deep image representations by
  inverting them}. In \bibinfo{booktitle}{\emph{Proceedings of the IEEE
  conference on computer vision and pattern recognition}}.
  \bibinfo{pages}{5188--5196}.
\newblock


\bibitem[\protect\citeauthoryear{Mohassel and Zhang}{Mohassel and
  Zhang}{2017}]%
        {mohassel2017secureml}
\bibfield{author}{\bibinfo{person}{Payman Mohassel} {and}
  \bibinfo{person}{Yupeng Zhang}.} \bibinfo{year}{2017}\natexlab{}.
\newblock \showarticletitle{Secureml: A system for scalable privacy-preserving
  machine learning}. In \bibinfo{booktitle}{\emph{2017 IEEE Symposium on
  Security and Privacy (SP)}}. IEEE, \bibinfo{pages}{19--38}.
\newblock


\bibitem[\protect\citeauthoryear{Ng et~al\mbox{.}}{Ng et~al\mbox{.}}{2011}]%
        {ng2011sparse}
\bibfield{author}{\bibinfo{person}{Andrew Ng} {et~al\mbox{.}}}
  \bibinfo{year}{2011}\natexlab{}.
\newblock \showarticletitle{Sparse autoencoder}.
\newblock \bibinfo{journal}{\emph{CS294A Lecture notes}} \bibinfo{volume}{72},
  \bibinfo{number}{2011} (\bibinfo{year}{2011}), \bibinfo{pages}{1--19}.
\newblock


\bibitem[\protect\citeauthoryear{Ohrimenko, Schuster, Fournet, Mehta, Nowozin,
  Vaswani, and Costa}{Ohrimenko et~al\mbox{.}}{2016}]%
        {ohrimenko2016oblivious}
\bibfield{author}{\bibinfo{person}{Olga Ohrimenko}, \bibinfo{person}{Felix
  Schuster}, \bibinfo{person}{C{\'e}dric Fournet}, \bibinfo{person}{Aastha
  Mehta}, \bibinfo{person}{Sebastian Nowozin}, \bibinfo{person}{Kapil Vaswani},
  {and} \bibinfo{person}{Manuel Costa}.} \bibinfo{year}{2016}\natexlab{}.
\newblock \showarticletitle{Oblivious multi-party machine learning on trusted
  processors}. In \bibinfo{booktitle}{\emph{25th $\{$USENIX$\}$ Security
  Symposium ($\{$USENIX$\}$ Security 16)}}. \bibinfo{pages}{619--636}.
\newblock


\bibitem[\protect\citeauthoryear{Paillier}{Paillier}{1999}]%
        {paillier1999public}
\bibfield{author}{\bibinfo{person}{Pascal Paillier}.}
  \bibinfo{year}{1999}\natexlab{}.
\newblock \showarticletitle{Public-key cryptosystems based on composite degree
  residuosity classes}. In \bibinfo{booktitle}{\emph{International Conference
  on the Theory and Applications of Cryptographic Techniques}}. Springer,
  \bibinfo{pages}{223--238}.
\newblock


\bibitem[\protect\citeauthoryear{Papernot, McDaniel, Goodfellow, Jha, Celik,
  and Swami}{Papernot et~al\mbox{.}}{2017}]%
        {papernot2017practical}
\bibfield{author}{\bibinfo{person}{Nicolas Papernot}, \bibinfo{person}{Patrick
  McDaniel}, \bibinfo{person}{Ian Goodfellow}, \bibinfo{person}{Somesh Jha},
  \bibinfo{person}{Z~Berkay Celik}, {and} \bibinfo{person}{Ananthram Swami}.}
  \bibinfo{year}{2017}\natexlab{}.
\newblock \showarticletitle{Practical black-box attacks against machine
  learning}. In \bibinfo{booktitle}{\emph{Proceedings of the 2017 ACM on Asia
  conference on computer and communications security}}. ACM,
  \bibinfo{pages}{506--519}.
\newblock


\bibitem[\protect\citeauthoryear{Salimans, Goodfellow, Zaremba, Cheung,
  Radford, and Chen}{Salimans et~al\mbox{.}}{2016}]%
        {salimans2016improved}
\bibfield{author}{\bibinfo{person}{Tim Salimans}, \bibinfo{person}{Ian
  Goodfellow}, \bibinfo{person}{Wojciech Zaremba}, \bibinfo{person}{Vicki
  Cheung}, \bibinfo{person}{Alec Radford}, {and} \bibinfo{person}{Xi Chen}.}
  \bibinfo{year}{2016}\natexlab{}.
\newblock \showarticletitle{Improved techniques for training gans}. In
  \bibinfo{booktitle}{\emph{Advances in neural information processing
  systems}}. \bibinfo{pages}{2234--2242}.
\newblock


\bibitem[\protect\citeauthoryear{{Samangouei}, {Kabkab}, and
  {Chellappa}}{{Samangouei} et~al\mbox{.}}{2018}]%
        {samangouei2018defense}
\bibfield{author}{\bibinfo{person}{Pouya {Samangouei}}, \bibinfo{person}{Maya
  {Kabkab}}, {and} \bibinfo{person}{Rama {Chellappa}}.}
  \bibinfo{year}{2018}\natexlab{}.
\newblock \showarticletitle{Defense-GAN: Protecting Classifiers Against
  Adversarial Attacks Using Generative Models}. In
  \bibinfo{booktitle}{\emph{ICLR 2018 : International Conference on Learning
  Representations 2018}}.
\newblock


\bibitem[\protect\citeauthoryear{{Sharif}, {Bhagavatula}, {Bauer}, and
  {Reiter}}{{Sharif} et~al\mbox{.}}{2016}]%
        {sharif2016accessorize}
\bibfield{author}{\bibinfo{person}{Mahmood {Sharif}}, \bibinfo{person}{Sruti
  {Bhagavatula}}, \bibinfo{person}{Lujo {Bauer}}, {and}
  \bibinfo{person}{Michael~K. {Reiter}}.} \bibinfo{year}{2016}\natexlab{}.
\newblock \showarticletitle{Accessorize to a Crime: Real and Stealthy Attacks
  on State-of-the-Art Face Recognition}. In
  \bibinfo{booktitle}{\emph{Proceedings of the 2016 ACM SIGSAC Conference on
  Computer and Communications Security}}. \bibinfo{pages}{1528--1540}.
\newblock


\bibitem[\protect\citeauthoryear{Shokri and Shmatikov}{Shokri and
  Shmatikov}{2015}]%
        {shokri2015privacy}
\bibfield{author}{\bibinfo{person}{Reza Shokri} {and} \bibinfo{person}{Vitaly
  Shmatikov}.} \bibinfo{year}{2015}\natexlab{}.
\newblock \showarticletitle{Privacy-preserving deep learning}. In
  \bibinfo{booktitle}{\emph{Proceedings of the 22nd ACM SIGSAC conference on
  computer and communications security}}. ACM, \bibinfo{pages}{1310--1321}.
\newblock


\bibitem[\protect\citeauthoryear{Xiao, Tsai, Sohn, Chandraker, and Yang}{Xiao
  et~al\mbox{.}}{2020}]%
        {xiao2019adversarial}
\bibfield{author}{\bibinfo{person}{Taihong Xiao}, \bibinfo{person}{Yi-Hsuan
  Tsai}, \bibinfo{person}{Kihyuk Sohn}, \bibinfo{person}{Manmohan Chandraker},
  {and} \bibinfo{person}{Ming-Hsuan Yang}.} \bibinfo{year}{2020}\natexlab{}.
\newblock \showarticletitle{Adversarial Learning of Privacy-Preserving and
  Task-Oriented Representations}.
\newblock \bibinfo{journal}{\emph{AAAI}} (\bibinfo{year}{2020}).
\newblock


\bibitem[\protect\citeauthoryear{Yang, Liu, Chen, and Tong}{Yang
  et~al\mbox{.}}{2019}]%
        {yang2019federated}
\bibfield{author}{\bibinfo{person}{Qiang Yang}, \bibinfo{person}{Yang Liu},
  \bibinfo{person}{Tianjian Chen}, {and} \bibinfo{person}{Yongxin Tong}.}
  \bibinfo{year}{2019}\natexlab{}.
\newblock \showarticletitle{Federated machine learning: Concept and
  applications}.
\newblock \bibinfo{journal}{\emph{ACM Transactions on Intelligent Systems and
  Technology (TIST)}} \bibinfo{volume}{10}, \bibinfo{number}{2}
  (\bibinfo{year}{2019}), \bibinfo{pages}{1--19}.
\newblock


\bibitem[\protect\citeauthoryear{Yao}{Yao}{1986}]%
        {yao1986generate}
\bibfield{author}{\bibinfo{person}{Andrew Chi-Chih Yao}.}
  \bibinfo{year}{1986}\natexlab{}.
\newblock \showarticletitle{How to generate and exchange secrets}. In
  \bibinfo{booktitle}{\emph{27th Annual Symposium on Foundations of Computer
  Science (sfcs 1986)}}. IEEE, \bibinfo{pages}{162--167}.
\newblock


\bibitem[\protect\citeauthoryear{Zhang and Yang}{Zhang and Yang}{2021}]%
        {zhang2021survey}
\bibfield{author}{\bibinfo{person}{Yu Zhang} {and} \bibinfo{person}{Qiang
  Yang}.} \bibinfo{year}{2021}\natexlab{}.
\newblock \showarticletitle{A survey on multi-task learning}.
\newblock \bibinfo{journal}{\emph{IEEE Transactions on Knowledge and Data
  Engineering}} (\bibinfo{year}{2021}).
\newblock


\end{thebibliography}


\end{document}